\documentclass[10pt,conference]{IEEEtran}
\usepackage{graphicx,subfigure}
\usepackage{amsfonts,amsmath,amssymb, mathrsfs}
\usepackage{multirow}
\usepackage{epic,eepic,eepicemu}
\usepackage{epsf}
\usepackage{epsfig}
\usepackage{graphics}
\usepackage{psfrag}

\newtheorem{theorem}{Theorem}
\newtheorem{definition}{Definition}

\newtheorem{discussion}{Discussion}
\newtheorem{claim}{Claim}

\newtheorem{lemma}{Lemma}

\newtheorem{corollary}{Corollary}
\newtheorem{remark}{Remark}

\begin{document}
%\input{Chead.tex}
% paper title
\title{On Marton's Inner Bound for the General Broadcast Channel}
\author{
\authorblockN{Amin Aminzadeh Gohari}
\authorblockA{EE Department\\
Sharif University of Technology\\
Tehran, Iran\\
Email: aminzadeh@sharif.edu} \and
\authorblockN{Abbas El Gamal}
\authorblockA{EE Department\\
Stanford University\\
Stanford, CA 94305, USA\\
Email: abbas@ee.stanford.edu}
 \and
\authorblockN{Venkat Anantharam}
\authorblockA{EECS Department\\
University of California\\
Berkeley, CA 94720, USA\\
Email: ananth@eecs.berkeley.edu}}
\maketitle

\begin{abstract}
We establish several new results on Marton's coding scheme and its corresponding inner bound on the capacity region of the general broadcast channel. We show
that unlike the Gaussian case, Marton's coding scheme without superposition coding is not optimal in general even for a degraded broadcast channel with no common message. We then establish properties of Marton's inner bound that help restrict the search space for computing the sum-rate. Next, we show that the inner bound is optimal along certain directions. Finally, we propose a coding scheme that may lead to a larger inner bound.
\end{abstract}

\section{Introduction}
In this paper, we consider the general two-receiver broadcast channel with an input alphabet ${\cal X}$, output alphabets ${\cal Y}$ and ${\cal Z}$, and conditional probability distribution function $q(y,z|x)$. The capacity region of this channel is defined as the set
of rate triples $(R_0, R_1, R_2)$ such that the sender $X$ can reliably communicate a common message at rate $R_0$ to both receivers and two private messages at rates $R_1$ and $R_2$ to receivers $Y$ and $Z$ respectively, see~\cite{CoverThomas} or \cite{CsiszárKörner}.
The capacity region of this channel is known for several special cases but unknown in general. The best known general inner bound to the capacity region is due to Marton \cite{Marton}\cite{LiangKramerPoor}.

In this paper, we study Marton's inner bound. Marton's inner bound for a general two-receiver discrete-memoryless broadcast channel is as follows:

{\em Marton's Inner bound \cite{Marton}\cite{CsiszárKörner}\cite{GelfandPinsker}\cite{LiangKramerPoor}}:
The union of non-negative rate triples $(R_0, R_1, R_2)$ satisfying the inequalities
  \begin{align}
    R_0+R_1&\leq I(UW;Y),\label{eqn:DefOfMKG2}\\
    R_0+R_2&\leq I(VW;Z),\label{eqn:DefOfMKG3}\\
    R_0+R_1+R_2&\leq
I(UW;Y)+I(V;Z|W)\nonumber\\&\qquad -I(U;V|W),\label{eqn:DefOfMKG4}\\
    R_0+R_1+R_2&\leq
I(U;Y|W)+I(VW;Z)\nonumber\\&\qquad-I(U;V|W),\label{eqn:DefOfMKG5}\\
    2R_0+R_1+R_2&\leq
I(UW;Y)+I(VW;Z)\nonumber\\&\qquad -I(U;V|W),\label{eqn:DefOfMKG6}
  \end{align}
for some random variables $(U,V,W,X,Y,Z)\sim p(u,v,w,x)q(y,z|x)$ constitutes an inner bound to the capacity region. Further to compute this region it suffices to consider $|U|, |V| \leq |X|, |W| \leq |X| + 4$ and assume that $X$ is a deterministic function of $(U, V, W)$ \cite{EvaluationMarton}.

In this paper we prove various results related to this inner bound.

\emph{Insufficiency of Marton's coding scheme without a superposition
variable:} Random variable $W$ corresponds to the ``superposition-coding" aspect of the bound, and the random variables $U$ and $V$ correspond to the ``Marton-coding" aspect of the bound. Necessity of the ``superposition-coding" aspect of the inner bound had previously been observed for a non-degraded broadcast channel \cite{JogNair}. For degraded channels, it is known that $W$ is unnecessary for achieving the capacity region of Gaussian broadcast channels (through dirty paper coding) \cite{WeingartenSteinbergShamai}. We show that, unlike in the Gaussian broadcast channel case, ``Marton's coding scheme" alone is not sufficient to achieve the capacity region of the general degraded broadcast channel.

\emph{Computing the sum rate in Marton's inner bound:} Given a broadcast channel $q(y,z|x)$ the maximum sum-rate achievable via Marton's strategy is given by
 \begin{align}
\max_{p(u,v,w,x)}& \min \{I(W;Y), I(W;Z) \} +\nonumber\\& I(U;Y|W) + I(V;Z|W) - I(U;V|W).
\label{eq:mib1l}
\end{align}
Further it suffices to consider $|U|, |V | \leq |X|, |W| \leq |X|+1$ and assume that $X$ is a deterministic function of $(U, V, W)$ \cite{EvaluationMarton}.
Note that $\min\{I(W;Y), I(W;Z) \}$ depends only on $p(w,x)$. The last three terms $I(U;Y|W) + I(V;Z|W) - I(U;V|W)$ can be written as $$\sum_{w}p(w)\big(I(U;Y|W=w) + I(V;Z|W=w) - I(U;V|W=w)\big).$$ Let us write the above optimization as follows:
 \begin{align}
\max_{p(w,x)}\bigg[& \min \{I(W;Y), I(W;Z) \} +\nonumber\\& \sum_{w}p(w)\max_{p(u,v|w,x)}\big[I(U;Y|W=w) + I(V;Z|W=w)\nonumber\\&~~~~~~~~~~~~~~~~~~~~~~~ - I(U;V|W=w)\big]\bigg].\nonumber
\end{align}
One can think of the maximization in the following way
 \begin{align*}
&\max_{p(w,x)}\min \{I(W;Y), I(W;Z) \} +\sum_{w}p(w)T(p(x|w))
\end{align*}
where $T(p(x))$ is the maximum of $I(U;Y)+ I(V;Z)-I(U;V)$ over all $p(u,v|x)$ where $X\sim p(x)$ and $X$ is a deterministic function of $(U, V)$.

It is shown in \cite{NairWangGeng} that the latter maximization problem concerning $T(p(x))$ has a remarkable solution for all binary input broadcast channels: it suffices to take $U=X$ and $V=constant$, or $V=X$ and $U=constant$. In other words, for all binary input broadcast channels we have
 \begin{align}I(U;Y) + I(V;Z) - I(U;V) \leq \max\{I(X;Y), I(X;Z)\}.\label{eq:binaryinequality}
\end{align}
To prove this, authors of \cite{NairWangGeng} consider different mappings from $\mathcal{U}\times \mathcal{V} \mapsto \mathcal{X}$. Because of the cardinality bound of two on $\mathcal{U}$ and $\mathcal{V}$, the authors argue that the XOR mapping (i.e. $X=U\oplus V\mod 2$) and the AND mapping (i.e. $X=U\wedge V$) cannot occur in any maximizer of $I(U;Y)+I(V;Z)-I(U;V)$.

We believe that finding the correct extension of equation (\ref{eq:binaryinequality}) to larger alphabets can be useful in (a) computing Marton's inner bound efficiently for a given channel, and (b) comparing the Marton inner bound with its multi-letter characterizations to see if Marton's inner bound is optimal or not (see \cite{ITApaper} for a discussion of this line of attack).

One of the main results of this part is to generalize to larger alphabets the statement that the XOR mapping cannot occur. We show that one cannot find distinct $u_0$, $u_1$ in $\mathcal{U}$, distinct $v_0$, $v_1$ in $\mathcal{V}$ and distinct $x_0$, $x_1$ in $\mathcal{X}$ such that $p(x_0|u_0,v_0)=p(x_0|u_1,v_1)=p(x_1|u_1,v_0)=p(x_1|u_0,v_1)=1$.
%
%Assume that $p^*(u,v,w,x)$ is an arbitrary joint distribution maximizing $\lambda I(W;Y)+ (1-\lambda)
%I(W;Z)+I(U;Y|W)+I(V;Z|W)-I(U;V|W)$, and having the largest value of $I(W;Y)+I(W;Z)$ among all maximizing joint distributions.

\emph{Optimality of Marton's inner bound along certain directions:} We compute the maximum of $\lambda_0R_0+\lambda_1R_1+\lambda_2R_2$ over all $(R_0, R_1, R_2)$ in the capacity region where $\lambda_0$, $\lambda_1$ and $\lambda_2$ are real numbers such that $\lambda_0\ge
\lambda_1+\lambda_2$. We observe that Marton's inner bound is tight along these directions.

\emph{An achievable region:} Since capacity is defined in the limit of large block length, it is natural to expect that optimal coding schemes have an invariant structure with respect
to shifts in time. This suggests that capacity should be expressed via a formula that has a \emph{fixed-point character}, namely it should involve joint distributions that are invariant under a time shift. Following this general idea, we propose a new inner bound for the capacity region. We don't know if the proposed inner bound is strictly better than Marton's inner bound.

The rest of the paper is organized as follows.
%In section
%\ref{Section:Definition}, we introduce the basic notation and
%definitions we use.
Section \ref{Section:MainResults}
contains the main results of the paper, and section
\ref{Section:Proofs} contains the proofs of these results, with
some of the details relegated to the appendices.
%----------------
%\section{Notation and Definitions}\label{Section:Definition}

\section{Main}\label{Section:MainResults}
Let $\mathcal{C}(q(y, z|x))$ denote the capacity region of the broadcast
channel $q(y,z|x)$, and $\mathcal{C}_{M}(q(y, z|x))$ denote Marton's inner bound for the
channel $q(y,z|x)$, defined in the introduction by equations (\ref{eqn:DefOfMKG2})-(\ref{eqn:DefOfMKG6}). The notation $X^{i}$ is used to denote the vector $(X_{1}, X_{2},...,
X_{i})$, and $X_{i}^{n}$ to denote $(X_{i}, X_{i+1},..., X_{n})$.

\subsection{Insufficiency of Marton's coding scheme
without a superposition variable}

In Marton's inner bound the auxiliary random variable $W$ corresponds to
the ``superposition-coding" aspect of the bound, and the random
variables $U$ and $V$ correspond to the ``Marton-coding" aspect of
the bound. When $R_0=0$ (private messages only) and $W=\emptyset$, Marton's inner bound reduces to
the the set of
non-negative rate pairs $(R_1,R_2)$ satisfying
  \begin{align}
    R_1&\leq I(U;Y|Q),\label{eqn:MartonWithoutW1}\\
    R_2&\leq I(V;Z|Q),\label{eqn:MartonWithoutW2}\\
    R_1+R_2&\leq
I(U;Y|Q)+I(V;Z|Q)-I(U;V|Q),\label{eqn:MartonWithoutW3}
  \end{align}
  for some random variables $(Q,U,V,X,Y,Z)\sim p(q)p(u,v,x|q)q(y,z|x)$.

It is known that this inner bound is
tight for Gaussian broadcast channels (through dirty paper coding), implying that $W$ is unnecessary for achieving the capacity region of this class of degraded broadcast channels \cite{WeingartenSteinbergShamai}. We show through an example that this is not the case in general.
\begin{claim}\label{Obs:Obs1} There are degraded broadcast channels for which
Marton's private message inner bound without $W$ is strictly contained in the capacity region of the channel (which is known to equal the Marton region with superposition variable in the case of degraded channels).
\end{claim}

\subsection{Computing the sum-rate for Marton's Inner Bound}
\subsubsection{Extensions of the binary inequality}
In this subsection we are concerned with the following maximization problem that is tightly related to the calculation of the sum rate for Marton's inner bound: given $p(x)$, maximize $I(U;Y)+I(V;Z)-I(U;V)$ over all $p(u,v|x)$ where $X$ is a function of $(U,V)$.

To state our main result we need the following two definitions:
\begin{definition}\label{Definition4} The input symbols $x_0$ and $x_1$ are said to be \emph{indistinguishable} by the channel if $q(y|x_0)=q(y|x_1)$ for all $y$, and $q(z|x_0)=q(z|x_1)$ for all $z$. A channel $q(y,z|x)$ is said to be \emph{irreducible} if no two of its inputs symbols are indistinguishable by the channel.
\end{definition}

\begin{definition}\label{Definition5} Let $\mathcal{U}=\{u_1, u_2,..., u_{|\mathcal{U}|}\}$, $\mathcal{V}=\{v_1, ..., v_{|\mathcal{V}|}\}$ be finite sets, and $\xi$ be a deterministic mapping from $\mathcal{U}\times \mathcal{V}$ to $\mathcal{X}$. One can represent the mapping by a table having $|\mathcal{U}|$ rows and $|\mathcal{V}|$ columns; the rows are indexed by $u_1, u_2, ...., u_{|\mathcal{U}|}$ and the columns are indexed by $v_1, v_2,..., v_{|\mathcal{V}|}$. In the cell $(i,j)$, we write $\xi(u_i, v_j)$, for the symbol $x$ that $(u_i, v_j)$ is being mapped to. The profile of the $i^{th}$ row is defined to be a vector of size $|\mathcal{X}|$ counting the number of occurrences of the elements of $\mathcal{X}$ in the $i^{th}$ row. In other words if $\mathcal{X}=\{x_1, x_2, ...,x_{|\mathcal{X}|}\}$, the $k^{th}$ element of the profile of the $i^{th}$ row is the number of times that $x_k$ shows up in the $i^{th}$ row of the table. The profile of the $j^{th}$ column is defined similarly. Define the profile of the table to be a vector of size $(|\mathcal{U}|+|\mathcal{V}|)|\mathcal{X}|$ formed by concatenating the profile vectors of the rows and the columns of the table. The profile vector of the mapping $\xi$ is denoted by $\overrightarrow{v_{\xi}}$.
\end{definition}

We now state the main result of this subsection.
\begin{theorem}\label{Thm:EB1} Take an arbitrary irreducible broadcast channel $q(y,z|x)$ where $q(y|x)>0, q(z|x)>0$ for all $x,y,z$. Fix some $p(x)$. Take any $p(u,v|x)$ maximizing $I(U;Y)+I(V;Z)-I(U;V)$ where $X$ is a function of $(U,V)$. Without loss of generality assume that $p(u)>0$ for all $u\in \mathcal{U}$, and $p(v)>0$ for all $v\in \mathcal{V}$. Let $x=\xi(u,v)$ denote the deterministic mapping from $\mathcal{U}\times \mathcal{V}$ to $\mathcal{X}$. Then all of the following conditions must hold:
\begin{itemize}
  \item $p(u,v)>0$, $p(u,y)>0$, and $p(v,z)>0$ for all $u, v, y$ and $z$.
  \item The profile vector of the mapping $\xi$, $\overrightarrow{v_{\xi}}$, cannot be written as
$$\sum_{t=1}^{M}\alpha_t\overrightarrow{v_{\xi_t}},$$
where $\xi_t$ (for $t=1,2,3,...,M$) are deterministic mappings from $\mathcal{U}\times \mathcal{V}$ to $\mathcal{X}$ not equal to $\xi$, and $\alpha_t$ are non-negative numbers adding up to one, i.e. $\sum_{t=1}^{M}\alpha_t=1$.
  \item Let the functions \begin{eqnarray*}&f_{u}:\mathcal{X}\rightarrow \mathbb{R}\mbox{ for every }u\in \mathcal{U},\\&g_{v}:\mathcal{X}\rightarrow \mathbb{R}\mbox{ for every }v\in \mathcal{V},\\&
\mbox{ and }h:\mathcal{X}\rightarrow \mathbb{R},\end{eqnarray*} be defined by
\begin{eqnarray*}&f_{u}(x)= \sum_{y}q(y|x)\log p(u,y),
\\&
g_{v}(x)=\sum_{z}q(z|x)\log p(v,z),
\end{eqnarray*}
$$h(x)=\min_{u'\in \mathcal{U},v'\in \mathcal{V}}\bigg(\log (p(u',v'))$$$$-f_{u'}(x)-g_{v'}(x)\bigg).$$
These definitions make sense because of the first bullet of this theorem. Then, for any $u$ and $v$, the following two equations hold:
    \begin{eqnarray*}&\log(p(u,v))=\max_x [f_{u}(x)+g_{v}(x)+h(x)],\end{eqnarray*}
    and
  \begin{eqnarray*}&p(x_0|u,v)=1 \mbox{ for some } x_0\in \mathcal{X} \Rightarrow \\& x_0\in argmax_x [f_{u}(x)+g_{v}(x)+h(x)].\end{eqnarray*}
\end{itemize}

\end{theorem}
\begin{discussion}
These constraints imply restrictions on the maximizers. The second bullet implies that one cannot find distinct $u_0$, $u_1$ in $\mathcal{U}$, distinct $v_0$, $v_1$ in $\mathcal{V}$ and distinct $x_0$, $x_1$ in $\mathcal{X}$ such that $p(x_0|u_0,v_0)=p(x_0|u_1,v_1)=p(x_1|u_1,v_0)=p(x_1|u_0,v_1)=1$.\footnote{Let the mapping $\xi_1$ be equal to $\xi$ except that $(u_0,v_0)$ and $(u_1,v_1)$ are mapped to $x_1$ (instead of $x_0$), and $(u_1,v_0)$ and $(u_0,v_1)$ are mapped to $x_0$ (instead of $x_1$). Figure \ref{MappingSameProfile} illustrates this. The mapping $\xi_1$ has the same profile vector as $\xi$. Thus we can write the original profile as a convex combination of other profiles (the condition in the displayed equation of the second bullet is violated for the choice of $M=1$, $\xi_1$ and $\alpha_1=1$). Thus the second bullet implies that it cannot happen. Similarly the mapping shown in Figure \ref{MappingSameProfile2} cannot occur because there is another mapping with the same profile.}
\begin{figure}
\centering
\includegraphics[width=65mm]{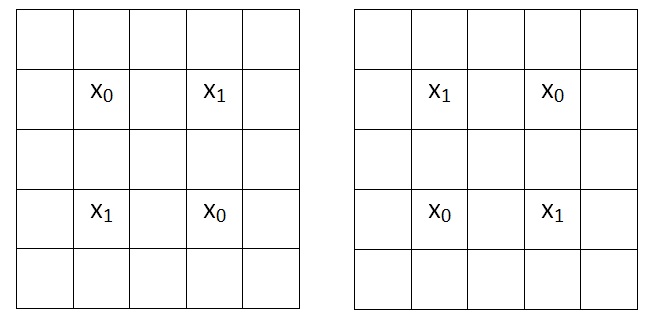}
\caption{If we have a mapping with the XOR structure, we can get another mapping with the same profile by switching $x_0$ and $x_1$ of four cells of the mappings. } \label{MappingSameProfile}
\end{figure}
\begin{figure}
\centering
\includegraphics[width=76mm]{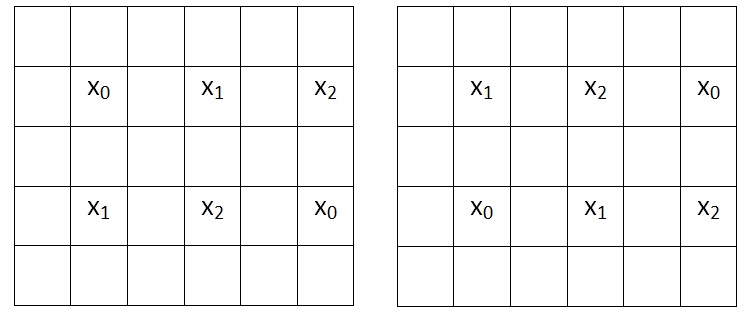}
\caption{Another mapping that cannot occur because one can find another mapping with the same profile.} \label{MappingSameProfile2}
\end{figure}

Next, assume that all we know about the mapping pattern is that $x_0=\xi(u_0,v_0)=\xi(u_1,v_1)$ for some $x_0$. Then the third bullet implies that $p(u_0,v_0)p(u_1,v_1)\leq p(u_1,v_0)p(u_0,v_1)$. This holds since \begin{eqnarray*}&\log p(u_0,v_0)+ \log p(u_1,v_1) =\\& f_{u_0}(x_0)+g_{v_0}(x_0)+h(x_0)+\\&f_{u_1}(x_0)+g_{v_1}(x_0)+h(x_0)=\\&
 f_{u_0}(x_0)+g_{v_1}(x_0)+h(x_0)+\\&
 f_{u_1}(x_0)+g_{v_0}(x_0)+h(x_0)\leq\\&
 \max_{x} f_{u_0}(x)+g_{v_1}(x)+h(x)+\\&
 \max_{x}f_{u_1}(x)+g_{v_0}(x)+h(x)= \\&
 \log p(u_0,v_1)+ \log p(u_1,v_0). \end{eqnarray*}
\end{discussion}

\subsubsection{Sum rate evaluation}
In this subsection we turn to evaluation of the whole sum-rate expression of Marton's inner bound (including the $W$ terms). We need the following definition:

For any $\lambda\in [0,1]$, let \begin{eqnarray*}&T_\lambda=\max_{p(u,v,w,x)}\big(\lambda I(W;Y)+ (1-\lambda) I(W;Z)+\\& I(U;Y|W)+I(V;Z|W)-I(U;V|W)\big).\end{eqnarray*}
Computing the sum-rate for Marton's inner bound is closely related to the above maximization problem for $\lambda\in [0,1]$:
\begin{claim}\label{Obs:ObsForMarton1}The maximum of the sum-rate for Marton's inner bound is equal to $\min_{\lambda\in [0,1]}T_\lambda$.
\end{claim}
Since the original submission of the conference version of this paper, some interesting properties of $T_\lambda$ such as its convexity in $\lambda$, and its connection to the outer bound and its factorization have been investigated in \cite{ITApaper} and \cite{ISITpaperChandra}. An alternative proof of Claim \ref{Obs:ObsForMarton1} using a theorem by Terkelsen is also reported in \cite{ITApaper}.

The main theorem of this section restricts the search space for computing $T_\lambda$. In this section, we only deal with broadcast channels $q(y,z|x)$ with strictly positive transition matrices, i.e. when $q(y|x)>0, q(z|x)>0$ for all $x,y,z$. In order to evaluate $T_\lambda$ when $q(y|x)$ or $q(z|x)$ become zero for some $y$ or $z$, one can use the continuity of $T_\lambda$ in $q(y,z|x)$ and take the limit of $T_\lambda$ for a sequence of channels with positive entries converging to the desired channel. The reason for dealing with this class of broadcast channels should become clear by the following lemma which is a corollary to the first bullet of Theorem \ref{Thm:EB1}.
\begin{lemma}\label{Lemma:LemmaForMarton1}Take an arbitrary broadcast channel $q(y,z|x)$ with strictly positive transition matrices (i.e. $q(y|x)>0, q(z|x)>0$ for all $x,y,z$). Let $p(u,v,w,x)$ be an arbitrary joint distribution maximizing $T_\lambda$ for some $\lambda\in [0,1]$ where $H(X|U,V,W)=0$. If $p(u,w)$ and $p(v,w)$ are positive for some triple $(u,v,w)$, then it must be the case that $p(u,v,w)>0$, $p(u,w,y)>0$ and $p(v,w,z)>0$ for all $y$ and $z$.
\end{lemma}
\begin{theorem}\label{Thm:ThmMarton1} Take an arbitrary irreducible broadcast channel $q(y,z|x)$ with strictly positive transition matrices. In computing $T_\lambda$ for some $\lambda\in [0,1]$, it suffices to take the maximum over auxiliary random variables $p(u,v,w,x)q(y,z|x)$ simultaneously satisfying the following constraints:
\begin{itemize}
  \item $|\mathcal{U}|\leq \min(|\mathcal{X}|, |\mathcal{Y}|),$ $|\mathcal{V}|\leq \min(|\mathcal{X}|, |\mathcal{Z}|),$ $|\mathcal{W}|\leq |\mathcal{X}|.$
  \item $H(X|UVW)=0$. Given $w$ where $p(w)>0$, we use $x=\xi^{(w)}(u,v)$ to denote the deterministic mapping from $\mathcal{U}_w\times \mathcal{V}_w$ to $\mathcal{X}$. Here $\mathcal{U}_w$ is the set of $u\in \mathcal{U}\mbox{ such that }p(u|w)>0$ and $\mathcal{V}_w$ is the set of $v\in \mathcal{V} \mbox{ such that }p(v|w)>0$.
  \item For arbitrary $w$ such that $p(w)>0$, the profile vector of the mapping $\xi^{(w)}$, $\overrightarrow{v_{\xi^{(w)}}}$, cannot be written as
$$\sum_{t=1}^{M}\alpha_t\overrightarrow{v_{\xi_t}},$$
where $\xi_t$ (for $t=1,2,3,...,M$) are deterministic mappings from $\mathcal{U}_w\times \mathcal{V}_w$ to $\mathcal{X}$ not equal to $\xi^{(w)}$, and $\alpha_t$ are non-negative numbers adding up to one, i.e. $\sum_{t=1}^{M}\alpha_t=1$.
  \item For arbitrary $w$ such that $p(w)>0$, let the functions \begin{eqnarray*}&f_{u,w}:\mathcal{X}\rightarrow \mathbb{R}\mbox{ for every }u\in \mathcal{U}_w,\\&g_{v,w}:\mathcal{X}\rightarrow \mathbb{R}\mbox{ for every }v\in \mathcal{V}_w,\\&
\mbox{ and }h_{w}:\mathcal{X}\rightarrow \mathbb{R},\end{eqnarray*} be defined by
\begin{eqnarray*}&f_{u,w}(x)= \sum_{y}q(y|x)\log p(uy|w),
\\&
g_{v,w}(x)=\sum_{z}q(z|x)\log p(vz|w),
\end{eqnarray*}
$$h_w(x)=\min_{u'\in \mathcal{U}_w,v'\in \mathcal{V}_w}\bigg(\log (p(u'v'|w))$$$$-f_{u',w}(x)-g_{v',w}(x)\bigg).$$
These definitions make sense because of Lemma \ref{Lemma:LemmaForMarton1}. Then, for any $u\in \mathcal{U}_w$ and $v\in \mathcal{V}_w$, the following two equations hold:
    \begin{eqnarray*}&\log(p(uv|w))=\max_x [f_{u,w}(x)+g_{v,w}(x)+h_w(x)],\end{eqnarray*}
    and
  \begin{eqnarray*}&p(x_0|u,v,w)=1 \mbox{ for some } x_0\in \mathcal{X} \Rightarrow \\& x_0\in argmax_x [f_{u,w}(x)+g_{v,w}(x)+h_w(x)].\end{eqnarray*}
  \item Given any $w$, random variables $U_w, V_w, X_w, Y_w, Z_w$ distributed according to $p(u,v,x,y,z|w)$ satisfy the following:
      %\begin{enumerate}
      \small
      \begin{eqnarray*}&I(\overline{U};Y_w)\geq I(\overline{U};V_wZ_w)\mbox{ for any }\overline{U}\rightarrow U_w \rightarrow V_wX_wY_wZ_w,
\\&I(\overline{V};Z_w)\geq I(\overline{V};U_wY_w)\mbox{ for any }\overline{V}\rightarrow V_w \rightarrow U_wX_wY_wZ_w.\end{eqnarray*}
\normalsize
\end{itemize}

\end{theorem}
\begin{discussion}
The first constraint imposes cardinality bounds on $|\mathcal{U}|$ and $|\mathcal{V}|$ that are better than those reported in \cite{EvaluationMarton}. \emph{However, we only claim the improved cardinality bounds for $T_\lambda$ and not the whole capacity region.} The second constraint is not new, and can be found in \cite{EvaluationMarton}. The other constraints are useful in restricting the search space due to the constraints imposed on $p(u,v,w,x)$. For instance, the third and fourth bullet restrict the set of possible mappings, as discussed in the previous subsection.
\end{discussion}

\subsection{Optimality along certain directions}
In order to state the main result of this section we need the following definition:
\begin{definition}\label{Definition3}\cite{Korner-Marton} Let $\mathcal{C}_{d_1}(q(y, z|x))$ and
$\mathcal{C}_{d_2}(q(y, z|x))$ denote the degraded message set capacity regions, i.e. when $R_1=0$ and $R_2=0$, respectively. The capacity region
$\mathcal{C}_{d_1}(q(y, z|x))$ is the set of of non-negative rate pairs
$(R_0, R_2)$ satisfying
\begin{align*}
    R_0&\leq I(W;Y), \\
    R_2&\leq I(X;Z|W),\\
 R_0+R_2&\leq I(X;Z),
  \end{align*}
for some random variables $(W,X,Y,Z)\sim p(w,x)q(y,z|x)$. The capacity region
$\mathcal{C}_{d_2}(q(y, z|x))$ is
defined similarly.
\end{definition}

We now state the main result of this section:
\begin{theorem}\label{Obs:Obs2} For a broadcast channel $q(y, z|x)$ and real numbers
$\lambda_0$, $\lambda_1$ and $\lambda_2$ such that $\lambda_0\ge
\lambda_1+\lambda_2$,
\begin{align*}
&\max_{(R_0,R_1,R_2) \in
\mathcal{C}(q(y,z|x))} (\lambda_0R_0+\lambda_1R_1+\lambda_2R_2) =\\&\max \{ \max_{(R_0,R_2)\in
\mathcal{C}_{d_1}(q(y,z|x))} (\lambda_0R_0+\lambda_2R_2) ,\\&
\max_{(R_0,R_1)\in
\mathcal{C}_{d_2}(q(y,z|x))} (\lambda_0R_0+\lambda_1R_1) \},
\end{align*}
where $\mathcal{C}_{d_1}(q(y,z|x))$ and
$\mathcal{C}_{d_2}(q(y,z|x))$ are the degraded message set capacity
regions for the given channel.
\end{theorem}
\begin{corollary}
The above observation essentially says that if $\lambda_0\geq \lambda_1+\lambda_2$, then \emph{a} maximum of
$\lambda_0R_0+\lambda_1R_1+\lambda_2R_2$ over triples $(R_0, R_1,
R_2)$ in the capacity region occurs when either $R_1=0$ or $R_2=0$.
\end{corollary}
\begin{remark} Since $\mathcal{C}_{d_1}(q(y,z|x))\cup\mathcal{C}_{d_1}(q(y,z|x))\subset \mathcal{C}_M(q(y,z|x))\subset \mathcal{C}(q(y,z|x))$, the above lemma implies that Marton's inner bound is tight along the
direction of such $(\lambda_0,\lambda_1,\lambda_2)$, i.e.
\begin{align*}
&\max_{(R_0,R_1,R_2)\in
\mathcal{C}(q(y,z|x))} (\lambda_0R_0+\lambda_1R_1+\lambda_2R_2) =\\&\quad \max_{(R_0,R_1,R_2)\in
\mathcal{C}_M(q(y,z|x))} (\lambda_0R_0+\lambda_1R_1+\lambda_2R_2), \\& \mbox{whenever }\lambda_0\geq \lambda_1+\lambda_2.
\end{align*}
%Based on numerical simulations for certain broadcast channels, we conjecture that the Nair-El Gamal outer bound \cite{Nair-El-Gamal} is also tight along the
%direction of any such $(\lambda_0,\lambda_1,\lambda_2)$. However if this conjecture turns out to be false,
%it would imply that the Nair-El Gamal outer bound is not tight.
\end{remark}

\subsection{An achievable region}
Since capacity is defined in the limit of large block length, it is natural to
expect that optimal coding schemes have an invariant structure with respect
to shifts in time. This suggests that capacity should be expressed
via a formula that has a fixed-point character, namely it should involve joint
distributions that are invariant under a time shift. The following theorem is
a proposed inner bound along these lines.

\begin{theorem}\label{Thm:Thm2} For a broadcast channel $q(y, z|x)$, consider
two i.i.d. copies
$(U_1,V_1,W_1)$ and $(U_2,V_2,W_2)$ and a conditional pmf
$r(x|u_1,v_1,w_1,u_2,v_2,w_2)$. Assume that $U_1, V_1, W_1,
U_2, V_2, W_2, X_1, X_2, Y_1, Y_2, Z_1, Z_2$ are distributed
according to
\begin{align*}
&p(u_1, v_1, w_1, u_2, v_2, w_2, x_1,
y_1, z_1, x_2, y_2,
z_2)=\\&r(u_1,v_1,w_1)r(u_2,v_2,w_2)\cdot\\&r(x_2|u_1,v_1,w_1,u_2,v_2,w_2)q(y_2,z_2|x_2)\cdot\\
&\tilde{r}(x_1|u_1,v_1,w_1)q(y_1,z_1|x_1),
\end{align*}
where $\tilde{r}(x|u,v,w)$ is defined as
\begin{align*}
&\sum_{u'\in \mathcal{U},v'\in \mathcal{V},w'\in \mathcal{W}}r(x|u',v',w',u,v,w)r(u',v',w').
\end{align*}

Then a rate triple $(R_0,R_1,R_2)$ is achievable if% \small
  \begin{align*}
    R_0, R_1, R_2&\geq 0,\\
    R_0+R_1&< I(U_2W_2;Y_1Y_2U_1W_1),\\
    R_0+R_2&< I(V_2W_2;Z_1Z_2V_1W_1),\\
    R_0+R_1+R_2&< I(V_2;Z_1Z_2V_1W_1|W_2)
\\&\quad + I(U_2W_2;Y_1Y_2U_1W_1)-I(U_2;V_2|W_2),\\
    R_0+R_1+R_2&<
I(U_2;Y_1Y_2U_1W_1|W_2)\\&\quad + I(V_2W_2;Z_1Z_2V_1W_1)-I(U_2;V_2|W_2),\\
    2R_0+R_1+R_2&<
I(U_2W_2;Y_1Y_2U_1W_1)\\&\quad + I(V_2W_2;Z_1Z_2V_1W_1)-I(U_2;V_2|W_2),
  \end{align*}\normalsize
  for some $U_1, V_1, W_1,U_2, V_2, W_2, X_1, X_2$ that satisfy the above conditions.
\end{theorem}
\begin{remark} The above inner bound reduces to Marton's inner bound if the conditional distribution
$r(x|u_1,v_1,w_1,u_2,v_2,w_2)=r(x|u_2,v_2,w_2)$, i.e. $U_1V_1W_1\rightarrow U_2V_2W_2\rightarrow X$ form a Markov chain.
\end{remark}
\section{Proofs}\label{Section:Proofs}
\begin{proof}[Proof of Claim \ref{Obs:Obs1}] Consider the degraded broadcast channel $p(y,z|x)=p(y|x)p(z|y)$, where the channel from $X$ to $Y$ is a BSC($0.3$)
and the channel from $Y$
to $Z$ is as follows: $p_{Z|Y} (0|0)=0.6$, $p_{Z|Y}(1|0)=0.4$,
$p_{Z|Y}(0|1)=0$, $p_{Z|Y}(1|1)=1$. We show that the private message capacity region for this channel is strictly larger than Marton's inner bound without $W$.

We first intuitively sketch outline of the proof: take a
non-negative real $\alpha$ and consider the maximum of $R_1+\alpha
R_2$ over the pairs $(R_1, R_2)$ in the capacity region. Since the
broadcast channel is degraded, the maximum is equal to
$\max_{V\rightarrow X\rightarrow YZ}I(X;Y|V)+\alpha I(V;Z)$. Since $
X\rightarrow Y\rightarrow Z$, when the weight of the degraded
receiver is less than or equal to $1$, an optimum $V$ will be equal
to a constant (corresponding to $R_2=0$). As we gradually increase
$\alpha$ beyond one, the optimum $V$ \emph{gradually} moves from a
constant random variable to $X$ (corresponding to $R_1=0$). Now, let
us consider the maximum of $R_1+\alpha R_2$ over the pairs $(R_1,
R_2)$ in Marton's inner bound without the auxiliary random variable $W$.
The latter maximum is equal to $I(U;Y)+\alpha I(V;Z)-I(U;V)$. When
$\alpha\leq 1$, it is optimum to take $U=X$, $V=$constant and
dedicate all the rate to the stronger receiver. Simulation results however indicate that as we increase $\alpha$ beyond one in
the problem of maximizing $I(U;Y)+\alpha I(V;Z)-I(U;V)$, $U=X$,
$V=$constant continues to be optimal up to a threshold. Beyond this
threshold, \emph{suddenly} $U=$constant, $V=X$ becomes the
optimizing choice, and stays as the optimizing choice afterwards.
In other words, unlike the gradual transition of the maximizing $V$
for the actual region,
 there is a \emph{sharp transition} in the maximizing $V$ for Marton's inner bound without $W$.

In the following, we provide a more detailed proof: the maximum of $R_1+2.4R_2$ over pairs $(R_1, R_2)$ in the capacity
region, is equal to $\max_{V\rightarrow X\rightarrow YZ}I(X;Y|V)+2.4
I(V;Z)$. Take the joint pmf of $p(v,x)$ to be as follows:
$P(V=0,X=0)=0$, $P(V=0,X=1)=0.41$, $P(V=1,X=0)=0.48$,
$P(V=1,X=1)=0.11$. For this choice of $p(v,x)$, $I(X;Y|V)+2.4
I(V;Z)=0.1229...$. Therefore the maximum of
$R_1+2.4R_2\ge 0.1229...$. The maximum of $R_1+2.4R_2$ over
Marton's inner bound without $W$ is equal to $\sup_{UV\rightarrow
X\rightarrow YZ}I(U;V)+2.4I(V;Z)-I(U;V)$. Using the perturbation method of \cite{EvaluationMarton}, one can bound the cardinality of $U$ and $V$ from above by $|\mathcal{X}|$, and further assume that $X$ is a deterministic function of $(U,V)$. This makes the domain compact, implying that the above supremum is
indeed a maximum.

Since $X$ is a
binary random variable, we need to search over binary random
variables $U$, $V$. Numerical simulations show that the maximum is
equal to $0.1215...<0.1229...$ and occurs when $X=V$ and
$U=constant$.
Therefore Marton's inner bound without $W$ is not tight for this broadcast channel.
\end{proof}

\begin{proof}[Proof of Claim \ref{Obs:ObsForMarton1}] In order to prove the observation, one needs to argue that the following exchange of max and min is legitimate:
\begin{eqnarray*}&\max_{p(u,v,w,x)}\min_{\lambda\in[0,1]}\lambda I(W;Y)+(1-\lambda)I(W;Z)+\\&I(U;Y|W)+I(V;Z|W)-I(U;V|W)=\\&
\min_{\lambda\in[0,1]}\max_{p(u,v,w,x)}\lambda I(W;Y)+(1-\lambda)I(W;Z)+\\&I(U;Y|W)+I(V;Z|W)-I(U;V|W).\end{eqnarray*}

Let $R_{Marton-Sum}$ denote the sum-rate for Marton's inner bound. We would like to show that
$R_{Marton-Sum}$ is equal to $\min_{0\leq \lambda \leq 1}T_\lambda$.

Let $\mathcal{D}$ be the union over all $p(u,v,w,x)$ of real pairs $(d_1, d_2)$ satisfying
\begin{eqnarray*}&d_1\leq I(W;Y)+I(U;Y|W)+I(V;Z|W)-I(U;V|W),\\&d_2\leq I(W;Z)+I(U;Y|W)+I(V;Z|W)-I(U;V|W).\end{eqnarray*}
We claim that this region is convex. Take two points $(d_1, d_2)$ and $(d'_1, d'_2)$ in the region. Corresponding to these are joint distributions $p(u_1,v_1,w_1,x_1)q(y_1,z_1|x_1)$ and $p(u_2,v_2,w_2,x_2)q(y_2,z_2|x_2)$. Take a uniform binary random variable $Q$ independent of all the previously defined random variables. Set $U=U_Q$, $V=V_Q$, $W=(Q,W_Q)$, $X=X_Q$, $Y=Y_Q$, $Z=Z_Q$. We will then have
\begin{eqnarray*}&I(W;Y)+I(U;Y|W)+I(V;Z|W)-I(U;V|W)=\\& I(W_Q,Q;Y_Q)+I(U_Q;Y_Q|W_Q,Q)+\\&I(V_Q;Z_Q|W_Q,Q)-I(U_Q;V_Q|W_Q,Q)\geq\\& I(W_Q;Y_Q|Q)+I(U_Q;Y_Q|W_Q,Q)+\\&I(V_Q;Z_Q|W_Q,Q)-I(U_Q;V_Q|W_Q,Q)=\\&
\frac{1}{2}\big(I(W_1;Y_1)+I(U_1;Y_1|W_1)+\\&I(V_1;Z_1|W_1)-I(U_1;V_1|W_1)\big)+\\&
\frac{1}{2}\big(I(W_2;Y_2)+I(U_2;Y_2|W_2)+\\&I(V_2;Z_2|W_2)-I(U_2;V_2|W_2)\big)\geq \\& \frac{1}{2}(d_1+d'_1).\end{eqnarray*}
Similarly,
\begin{eqnarray*}&I(W;Z)+I(U;Y|W)+I(V;Z|W)-I(U;V|W)\geq \\& \frac{1}{2}(d_2+d'_2).\end{eqnarray*}
Thus, the point $(\frac{1}{2}(d_1+d'_1),\frac{1}{2}(d_2+d'_2))$ is in the region. Thus, $\mathcal{D}$ is convex.

Next, note that the point $(R_{Marton-Sum}, R_{Marton-Sum})$ is in $\mathcal{D}$. We claim that it is a boundary point of $\mathcal{D}$. If it is an interior point, there must exist $\epsilon>0$ such that $(R_{Marton-Sum}+\epsilon, R_{Marton-Sum}+\epsilon)$ is in $\mathcal{D}$. This implies the existence of some $p(u,v,w,x)$ where
\begin{eqnarray*}&R_{Marton-Sum}+\epsilon\leq\\& I(W;Y)+I(U;Y|W)+I(V;Z|W)-I(U;V|W),\\&R_{Marton-Sum}+\epsilon\leq \\& I(W;Z)+I(U;Y|W)+I(V;Z|W)-I(U;V|W).\end{eqnarray*}
This implies that \begin{eqnarray*}&R_{Marton-Sum}+\epsilon\leq \\&\min(I(W;Y),I(W;Z))+I(U;Y|W)+\\&I(V;Z|W)-I(U;V|W)\end{eqnarray*} for some $p(u,v,w,x)$, which is a contradiction.

Using the supporting hyperplane theorem and the fact that $\mathcal{D}$ is convex and closed, one can conclude that there exists a supporting hyperplane to $\mathcal{D}$ at the boundary point $(R_{Marton-Sum}, R_{Marton-Sum})$. We claim that this supporting hyperplane must have the equation $\lambda^* d_1+(1-\lambda^*) d_2 = T(\lambda^*)$ for some $\lambda^* \in[0,1]$. The proof is as follows: any supporting hyperplane has the formula $\lambda^* d_1+(1-\lambda^*) d_2 = k$ for some real $\lambda^*$ and real $k$. We claim that $\lambda^*$ must be in $[0,1]$ and $k=T(\lambda^*)$. Assume that for instance $\lambda^*<0$. We know that $\mathcal{D}$ must be entirely contained in one of the two closed half-spaces determined by the hyperplane. Note that the points $(0,0)$, $(-\infty,0)$ and $(0,-\infty)$ are in $\mathcal{D}$ (take $p(u,v,w,x)$ satisfying $I(U;V|W)=0$ in the definition of $\mathcal{D}$). The value of  $\lambda^* d_1+(1-\lambda^*) d_2$ at these points is equal to $0$, $+\infty$ and $-\infty$ respectively. Thus, $\mathcal{D}$ cannot possibly be entirely contained in one of the two closed half-spaces determined by the hyperplane. Similarly the case $1-\lambda^*<0$ can be refuted. Therefore $\lambda^*$ must be in $[0,1]$. Since the points $(-\infty,0)$ and $(0,-\infty)$ are in $\mathcal{D}$, the half-space determined by the hyperplane that contains $\mathcal{D}$ is the one determined by the equation $\lambda^* d_1+(1-\lambda^*) d_2 \leq k$ for some $k$. Since the half-space has at least one point of $\mathcal{D}$, the value of $k$ must be equal to $\max_{(d_1, d_2)\in \mathcal{R}}\lambda^* d_1+(1-\lambda^*) d_2$. The latter is equal to $T(\lambda^*)$. Thus, the supporting hyperplane at the boundary point $(R_{Marton-Sum}, R_{Marton-Sum})$ has the equation $\lambda^* d_1+(1-\lambda^*) d_2 = T(\lambda^*)$ for some $\lambda^* \in[0,1]$.

Since $(R_{Marton-Sum}, R_{Marton-Sum})$ lies on this hyperplane, $\lambda^* R_{Marton-Sum}+(1-\lambda^*) R_{Marton-Sum} = T(\lambda^*)$ implies that $R_{Marton-Sum} = T(\lambda^*)$ for some $\lambda^* \in[0,1]$. Therefore $$\min_{0\leq \lambda \leq 1}T_\lambda\leq R_{Marton-Sum}.$$
On the other hand, for every $\lambda$, $T_\lambda\geq R_{Marton-Sum}$. Therefore
$$\min_{0\leq \lambda \leq 1}T_\lambda\geq R_{Marton-Sum}.$$
\end{proof}

\begin{proof}[Proof of Theorem \ref{Thm:EB1}] We begin by proving the \textbf{first bullet}. $p(u,y)$ is positive for all $u, y$ because there must exist some $x$ such that $p(u,x)>0$. Since the transition matrices have positive entries and $p(u,y)\geq p(u,x)q(y|x)$, $p(u,y)$ will be positive for all $y$. A similar argument proves that $p(v,z) > 0$
for all $v,z$. Next assume that $p(u,v)=0$ for some $(u,v)$. Take some $u', v'$ such that $p(u', v')>0$. Let us reduce $p(u',v')$ by $\epsilon$ and increase $p(u,v)$ by $\epsilon$. Furthermore, have $(u,v)$ mapped to the same $x$ that $(u',v')$ was mapped to; this ensures that the marginal distribution of $X$ is preserved. One can write
\begin{eqnarray*}&I(U;Y)+I(V;Z)-I(U;V)=\\&
H(Y)+H(Z)+H(UV)-H(UY)-H(VZ).\end{eqnarray*}
The only change in this expression comes from the change in $H(UV)-H(UY)-H(VZ)$. The derivative of $H(UV)$ with respect to $\epsilon$, at $\epsilon=0$, will be infinity. But the derivative of $H(UY)$ and $H(VZ)$ will be finite since $p(u,y)$, $p(u',y)$, $p(v,z)$ and $p(v',z)$ are positive for all $y$ and $z$. So, the first derivative of $H(UV)-H(UY)-H(VZ)$ with respect to $\epsilon$, at $\epsilon=0$, will be positive. This is a contradiction since $p(u,v|x)$ was assumed to maximize $I(U;Y)+I(V;Z)-I(U;V)$.

We now prove the \textbf{second bullet}. Assume that $\mathcal{U}=\{u_1, u_2, ...., u_{|\mathcal{U}|}\}$ and $\mathcal{V}=\{v_1, v_2, ...., v_{|\mathcal{V}|}\}$. Let $\pi_{i,j}=p(u_i,v_j)$ for $i=1,...,|\mathcal{U}|$, $j=1,..., |\mathcal{V}|$. From the first bullet we know that $\pi_{i,j}>0$ for all $i$ and $j$. Let $\overline{\epsilon}=\min_{i,j}\pi_{i,j}$. Take some $\epsilon \in (0, \overline{\epsilon})$. Let $x =\xi_0(u,v)$ denote the deterministic mapping from ${\cal U} \times {\cal V}$ to ${\cal X}$.

We prove the statement by contradiction. Assume that $$\overrightarrow{v_{\xi_0}}=\sum_{t=1}^{M}\alpha_t\overrightarrow{v_{\xi_t}},$$
for some mappings $\xi_t$ ($t=1,2,..,M$) distinct from $\xi_0$ and non-negative numbers $\alpha_t$ adding up to one.

Let random variables $T_{i,j}$ (for $i=1,...,|\mathcal{U}|$, $j=1,2,3,..., |\mathcal{V}|$) be $M+1$-ary random variables mutually independent of each other, and of $U,V,X,Y,Z$ satisfying:
\begin{itemize}
  \item $p(T_{i,j}=0)=1-\frac{\epsilon}{\pi_{i,j}}$,
  \item $p(T_{i,j}=1)=\frac{\epsilon}{\pi_{i,j}}\alpha_1$,
  \item $p(T_{i,j}=2)=\frac{\epsilon}{\pi_{i,j}}\alpha_2$,
  \item $p(T_{i,j}=3)=\frac{\epsilon}{\pi_{i,j}}\alpha_3$,
  \item ...
  \item $p(T_{i,j}=M)=\frac{\epsilon}{\pi_{i,j}}\alpha_M$.
\end{itemize}

Let $\widetilde{X}$ be defined as follows:
\begin{itemize}
  \item On the event that $(U,V)=(u_i,v_j)$, let $\widetilde{X}$ be equal to $\xi_{T_{i,j}}(u_i,v_j)$. In other words, if $T_{i,j}=0$, $\widetilde{X}$ is equal to $\xi_{0}(u_i,v_j)$; if $T_{i,j}=1$, $\widetilde{X}$ is equal to $\xi_{1}(u_i,v_j)$, etc.
\end{itemize}

We claim that $p(\widetilde{X}=x|U=u_i)=p(X=x|U=u_i)$ for all $i=1,2,3,...,|\mathcal{U}|$ and $x$; and similarly $p(\widetilde{X}=x|V=v_j)=p(X=x|V=v_j)$ for all $j=1,2,3,...,|\mathcal{V}|$ and $x$. This is proved in Appendix \ref{Appndx:E}. Note that the above property implies that $\widetilde{X}$ and $X$ have the same marginal distributions.

Let $\widetilde{Y}$ and $\widetilde{Z}$ be defined such that $UV(T_{i,j})_{i:1,2,..,j=1,2,..}\rightarrow \widetilde{X}\rightarrow \widetilde{Y}\widetilde{Z}$, and the conditional law of $\widetilde{y}$ and $\widetilde{z}$ given $\widetilde{x}$ is the same as $q(y,z|x)$. Here $(T_{i,j})_{i:1,2,..,j=1,2,..}$ denotes the collection of all $T_{i,j}$ for all $i$ and $j$.

Without loss of generality, let us assume $\alpha_1\neq 0$. Since the mapping $\xi_0(\cdot, \cdot)$ is not equal to $\xi_1(\cdot, \cdot)$, there must exist $(i,j)$ such that $\xi_0(u_i, v_j)\neq \xi_1(u_i, v_j)$. Let us label the input symbol $\xi_0(u_i, v_j)$ by $x_0$, and the input symbol $\xi_1(u_i, v_j)$ by $x_1$. We know that the channel is irreducible. Let us then assume that there is some $y$ such that $q(y|x_0)\neq q(y|x_1)$; the proof for the case when there is some $z$ such that $q(z|x_0)\neq q(z|x_1)$ is similar. Let $\widetilde{U}=(U, T_{i,j})$ and $\widetilde{V}=V$.

Since $p(\widetilde{X}=x|U=u)=p(X=x|U=u)$ for all $u$ and $x$, and $p(\widetilde{X}=x|V=v)=p(X=x|V=v)$ for all $v$ and $x$, we have
\begin{itemize}
  \item $I(U;\widetilde{Y})=I(U;Y)$,
  \item $I(V;\widetilde{Z})=I(V;Z)$.
\end{itemize}
Therefore $I(\widetilde{V};\widetilde{Z})=I(V;Z)$ and $I(\widetilde{U};\widetilde{Y})=I(U;Y)+I(T_{i,j};\widetilde{Y}|U)$. Furthermore since $T_{i,j}$ is independent of $U,V$, we have $I(\widetilde{U};\widetilde{V})=I(U;V)$. Therefore \small
\begin{eqnarray*}&I(\widetilde{U};\widetilde{Y})+I(\widetilde{V};\widetilde{Z})-I(\widetilde{U};\widetilde{V})-\big(I(U;Y)+I(V;Z)-I(U;V)\big)\\&
=I(T_{i,j};\widetilde{Y}|U).\end{eqnarray*}\normalsize

Since $p(u,v,x)$ was maximizing $I(U;Y)+I(V;Z)-I(U;V)$ under the fixed marginal distribution on $x$, we must have $I(T_{i,j};\widetilde{Y}|U)=0$. Therefore $I(T_{i,j};\widetilde{Y}|U=u_i)=0$ holds as well.

In Appendix \ref{Appndx:C01}, we prove that the following are true
$$p(\widetilde{X}=x_0|U=u_i, T_{i,j}=0)
\neq p(\widetilde{X}=x_0|U=u_i, T_{i,j}=1),$$
$$p(\widetilde{X}=x_1|U=u_i, T_{i,j}=0)\neq p(\widetilde{X}=x_1|U=u_i, T_{i,j}=1).$$ But for any $x\notin\{x_0,x_1\}$,
$$p(\widetilde{X}=x|U=u_i, T_{i,j}=0)= p(\widetilde{X}=x|U=u_i, T_{i,j}=1).$$

Remember that we assumed that there is some $y$ such that $q(y|x_0)\neq q(y|x_1)$. In Appendix \ref{Appndx:C02}, we show that
$$p(\widetilde{Y}=y|U=u_i, T_{i,j}=0)
\neq p(\widetilde{Y}=y|U=u_i, T_{i,j}=1).$$
This implies that $\widetilde{Y}$ and $T_{i,j}$ are not conditionally independent given $U=u_i$. Therefore $I(T_{i,j};\widetilde{Y}|U=u_i)\neq 0$ which is a contradiction.

We now prove the \textbf{third bullet}. The proof begins by noting that the definition of $h(x)$ implies that for any $(u,v,x)$, \begin{eqnarray*}&h(x)\leq \log (p(u,v))-f_u(x)-g_v(x).\end{eqnarray*} Therefore, for any $(u,v,x)$,
\begin{eqnarray*}&\log (p(u,v))\geq f_u(x)+g_v(x)+h(x).\end{eqnarray*} Thus,
\begin{eqnarray}&\log (p(u,v))\geq \max_{x}\big(f_u(x)+g_v(x)+h(x)\big).\label{eqn:appendixB1}\end{eqnarray}

Note that the first partial derivative of $H(UV)-H(UY)-H(VZ)$ with respect to $p(u,v,x)$ is proportional to \begin{eqnarray*}&-\log p(u,v) -1+\sum_{y}q(y|x)\log p(u,y)+1+\\&\sum_{z}q(z|x)\log p(v,z)+1=\\&
-\log p(u,v)+f_{u}(x)+g_{v}(x)+1.\end{eqnarray*}
Assume that the triple $(u,v,x)$ is such that $p(u,v,x)>0$. Take some arbitrary $u'$ and $v'$. Reducing $p(u,v,x)$ by a small $\epsilon$ and increasing $p(u',v',x)$ by $\epsilon$ does not affect the marginal distribution of $X$ and hence should not increase the expression $H(UV)-H(UY)-H(VZ)$. Therefore the first derivative of $H(UV)-H(UY)-H(VZ)$ with respect to $p(u,v,x)$ must be greater than or equal to the first derivative of $H(UV)-H(UY)-H(VZ)$ with respect to $p(u',v',x)$. Thus, \begin{eqnarray*}&-\log p(u,v)+f_{u}(x)+g_{v}(x)+1\geq\\& -\log p(u',v')+f_{u'}(x)+g_{v'}(x)+1.\end{eqnarray*}
In other words, for any arbitrary $u'$ and $v'$, we have
\begin{eqnarray*}&\log p(u,v)-f_{u}(x)-g_{v}(x)\leq\\& \log p(u',v')-f_{u'}(x)-g_{v'}(x).\end{eqnarray*}
Therefore
\begin{eqnarray*}&\log p(u,v)-f_{u}(x)-g_{v}(x)\leq\\& \min_{u',v'}\log p(u',v')-f_{u'}(x)-g_{v'}(x)=h(x).\end{eqnarray*}
Thus, $\log p(u,v)\leq f_{u}(x)+g_{v}(x)+h(x)$ whenever $p(u,v,x)>0$. This together with equation (\ref{eqn:appendixB1}) imply that
\begin{eqnarray*}&\log(p(u,v))=\max_x f_{u}(x)+g_{v}(x)+h(x),\end{eqnarray*}
and \begin{eqnarray*}&p(x_0|u,v)=1 \mbox{ for some } x_0\in \mathcal{X} \Rightarrow \\& x_0\in argmax_x f_{u}(x)+g_{v}(x)+h(x).\end{eqnarray*}

\end{proof}

\begin{proof}[Proof of Lemma \ref{Lemma:LemmaForMarton1}] This is a consequence of bullet one of Theorem \ref{Thm:EB1}.
\end{proof}

\begin{proof}[Proof of Theorem \ref{Thm:ThmMarton1}] From the set of pmfs $p(u,v,w,x)$ that maximize the expression $\lambda I(W;Y)+(1-\lambda)I(W;Z)+I(U;Y|W)+I(V;Z|W)-I(U;V|W)$, let $p_0(u,v,w,x)$ be the one that achieves the
largest value of $I(W;Y)+I(W;Z)$. In Appendix \ref{Appndx:B}, we prove that one can find $p(\widehat{u},\widehat{v},\widehat{w},\widehat{x})$ such that
\begin{itemize}
  \item $\lambda I(W;Y)+(1-\lambda)I(W;Z)+I(U;Y|W)+I(V;Z|W)-I(U;V|W)$ is equal to $\lambda I(\widehat{W};\widehat{Y})+(1-\lambda)I(\widehat{W};\widehat{Z})+I(\widehat{U};\widehat{Y}|\widehat{W})+I(\widehat{V};\widehat{Z}|\widehat{W})-I(\widehat{U};\widehat{V}|\widehat{W})$,
  \item $I(W;Y)+I(W;Z)$ is equal to $ I(\widehat{W};\widehat{Y})+I(\widehat{W};\widehat{Z})$,
  \item $|\mathcal{\widehat{U}}|\leq \min(|\mathcal{X}|, |\mathcal{Y}|),$
  \item $|\mathcal{\widehat{V}}|\leq \min(|\mathcal{X}|, |\mathcal{Z}|),$
  \item $|\mathcal{\widehat{W}}|\leq |\mathcal{X}|,$
  \item $H(\widehat{X}|\widehat{U}\widehat{V}\widehat{W})=0.$
\end{itemize}
Thus the constraints in the first and second bullets are satisfied by $p(\widehat{u},\widehat{v},\widehat{w},\widehat{x})$. The second and third bullets of Theorem \ref{Thm:EB1} imply that $p(\widehat{u},\widehat{v},\widehat{w},\widehat{x})$ will automatically satisfy the third and fourth bullet of Theorem \ref{Thm:ThmMarton1}. In Appendix \ref{Appndx:D}, we show that the fifth bullet of Theorem \ref{Thm:ThmMarton1} holds for any joint distribution that maximizes the expression $\lambda I(W;Y)+(1-\lambda)I(W;Z)+I(U;Y|W)+I(V;Z|W)-I(U;V|W)$, and at the same time has the largest possible value of $I(W;Y)+I(W;Z)$. Thus it must also hold for $p(\widehat{u},\widehat{v},\widehat{w},\widehat{x})$.
\end{proof}

\begin{proof}[Proof of Theorem \ref{Obs:Obs2}] It suffices to show that
\begin{align*}
&\max_{(R_0,R_1,R_2)\in
\mathcal{C}(q(y,z|x))} (\lambda_0R_0+\lambda_1R_1+\lambda_2R_2) \leq
\\&
\max \{ \max_{(R_0,R_2)\in
\mathcal{C}_{d_1}(q(y,z|x))} (\lambda_0R_0+\lambda_2R_2),\\&
\max_{(R_0,R_1)\in
\mathcal{C}_{d_2}(q(y,z|x))} (\lambda_0R_0+\lambda_1R_1)\}.
\end{align*}
The key step is to show that if $(R_0, R_1, R_2)$ is in the
capacity region of a broadcast channel, then
$(R_0+\min\{ R_1, R_2\}, R_1-\min \{R_1, R_2\}, R_2-\min \{R_1,
R_2\} )$ is also in the capacity region. Since $\lambda_0\geq \lambda_1+\lambda_2$, we then
have that $\lambda_0 (R_0+\min \{R_1,
R_2\} )+\lambda_1 (R_1-\min \{R_1,
R_2\} )+\lambda_2(R_2-\min \{R_1, R_2\})\ge \lambda_0R_0+\lambda_1R_1+\lambda_2R_2$, so at the maximum we must have $\min(R_1, R_2)=0$. One can prove this
property using the result of Willems \cite{Willems}, which shows that the \emph{maximal} probability of error capacity
region is equal to the \emph{average} probability of error capacity region. Willems's proof of his result, however, is rather involved. Instead, we
provide a simple direct proof. Consider an
arbitrary code $(M_0, M_1, M_2, X^n, \epsilon)$. We show that
\begin{align*}
&\frac{\lambda_0}{n}H(M_0)+\frac{\lambda_1}{n}H(M_1)+\frac{\lambda_2}{n}H(M_2)-O(\epsilon)\leq\\& \max (\max_{(R_0,R_2)\in
\mathcal{C}_{d_1}(q(y,z|x))}\lambda_0R_0+\lambda_2R_2,\\&
\max_{(R_0,R_1)\in
\mathcal{C}_{d_2}(q(y,z|x))}\lambda_0R_0+\lambda_1R_1),
\end{align*}
where $O(\epsilon)$ denotes a constant (depending only on $|\mathcal{X}|$, $|\mathcal{Y}|$, $|\mathcal{Z}|$) times $\epsilon$. %CLARIFY THIS

Assume without loss of generality that $H(M_2)\leq H(M_1)$, i.e.
$R_2\leq R_1$. Let $\widehat{W}=M_0M_2$,
$\widehat{X}=X^n$, $\widehat{Y}=Y^n$, $\widehat{Z}=Z^n$. Note that
$q(\widehat{y}, \widehat{z}|\widehat{x})$ is the $n$-fold version of
$q(y,z|x)$. Let us look at $\mathcal{C}_{d_1}(q(\widehat{y},
\widehat{z}|\widehat{x}))$, evaluated at the joint pmf
$p(\widehat{w}, \widehat{x})$:
\[
  \begin{array}{lll}
    \widehat{R_0}&\leq I(\widehat{W};\widehat{Z}), \\
    \widehat{R_1}&\leq I(\widehat{X};\widehat{Y}|\widehat{W}), \\
    \widehat{R_0}+\widehat{R_1}&\leq
I(\widehat{X};\widehat{Y}).\\
  \end{array}
\]
Note that, by Fano's inequality,
\begin{align*}
&I(\widehat{W};\widehat{Z})=I(M_0M_2;Z^n)=H(M_0)+H(M_2)-O(n\epsilon),\\
&I(\widehat{X};\widehat{Y}|\widehat{W})=I(X^n;Y^n|M_0M_2)=
H(M_1)-O(n\epsilon),\\&I(\widehat{X};\widehat{Y})= I(X^n;Y^n)\ge
H(M_0)+H(M_1)-O(n\epsilon).
\end{align*}
 Therefore
$\widehat{R_0}=H(M_0)+H(M_2)-O(n\epsilon)=n(R_0+R_2)-O(n\epsilon)$
and $\widehat{R_1}=H(M_1)-H(M_2)=n(R_1-R_2)-O(n\epsilon)$ is in
$\mathcal{C}_{d_1}(q(\widehat{y}, \widehat{z}|\widehat{x}))$. Since
$q(\widehat{y}, \widehat{z}|\widehat{x})$ is the $n$-fold version of
$q(y,z|x)$ and $\mathcal{C}_{d_1}(q(\widehat{y},
\widehat{z}|\widehat{x}))$ is the degraded message set capacity
region for $q(\widehat{y}, \widehat{z}|\widehat{x})$, we must have:
$\mathcal{C}_{d_1}(q(\widehat{y}, \widehat{z}|\widehat{x}))=n\cdot
\mathcal{C}_{d_1}(q(y,z|x)),$ where the multiplication here is pointwise. Thus, $(\frac{\widehat{R_0}}{n}, \frac{\widehat{R_1}}{n})\in
\mathcal{C}_{d_1}(q(y,z|x))$. We can complete the proof by
letting $\epsilon \to 0$, and conclude that $(R_0+R_2,
R_1-R_2, 0)\in \mathcal{C}_{d_1}(q(y,z|x))$, and thus also in the capacity region.
\end{proof}

\begin{proof}[Proof of Theorem \ref{Thm:Thm2}]
Consider a natural number $n$, and define the super symbols
$\tilde{X}=X_1X_2...X_{n}$, $\tilde{Y}=Y_1Y_2...Y_{n}$,
$\tilde{Z}=Z_1Z_2...Z_{n}$ representing $n$-inputs and $n$-outputs of
the product broadcast channel
$$q^{n}(y_1y_2...y_{n}, z_1z_2...z_{n}|x_1x_2...x_{n})=\prod_{i=1}^{n}q(y_i, z_i|x_i).$$
Since the capacity region of the product channel $q^{n}(\tilde{y},
\tilde{z}|\tilde{x})$ is $n$ times the capacity region of $q(y,
z|x)$, we have $\frac{1}{n}\mathcal{C}_{M}(q^{n}(y_1y_2...y_{n},
z_1z_2...z_{n}|x_1x_2...x_{n}))\subset \mathcal{C}(q(y, z|x))$.
Given an \emph{arbitrary} joint pmf $p(u^n,v^n,w^n,x^n)$,
one can then show that the following region is an inner bound to
$\mathcal{C}(q(y, z|x))$:\small
  \begin{align}
  R_0, R_1, R_2&\geq 0,\nonumber\\
    R_0+R_1&\leq \frac{1}{n}I(U^nW^n;Y^n),\label{eqna1} \\
    R_0+R_2&\leq \frac{1}{n}I(V^nW^n;Z^n),\label{eqna2} \\
    R_0+R_1+R_2&\leq
\frac{1}{n}\big[I(U^nW^n;Y^n)+I(V^n;Z^n|W^n)\nonumber\\&\quad\quad\quad-I(U^n;V^n|W^n)\big],\label{eqna3}\\
 R_0+R_1+R_2&\leq
\frac{1}{n}\big[I(U^n;Y^n|W^n)+I(V^nW^n;Z^n)\nonumber\\&\quad\quad\quad-I(U^n;V^n|W^n)\big], \label{eqna4}\\
 2R_0+R_1+R_2&\leq
\frac{1}{n}\big[I(U^nW^n;Y^n)+I(V^nW^n;Z^n)\nonumber\\&\quad\quad\quad-I(U^n;V^n|W^n)\big],\label{eqna5}
  \end{align}\normalsize
where $U^n,V^n,W^n,X^n,Y^n,Z^n$ are distributed according to
$p(u^n,v^n,w^n,x^n)q(y^n,z^n|x^n)$. Clearly if we assume that
$(U^n,V^n,W^n,X^n)$ is $n$ i.i.d. copies of $p(u,v,w,x)$ we get back
the one-letter version of Marton's inner bound.

Assume that
\begin{align*}&p(u^n,v^n,w^n)=\prod_{i=1}^nr(u_i,v_i,w_i).\end{align*} Note that $U_i,V_i,W_i$ are i.i.d. copies of $(U,V,W)$ distributed
according to $r(u,v,w)$. We further use the given conditional law
$r(x|u_1,v_1,w_1,u_2,v_2,w_2)$ to define the joint distribution of
$X^n$ given $U^n, V^n, W^n$ as
\begin{align*}&p(x_2^n|u^n,v^n,w^n)=\prod_{i=2}^n r(x_i|u_{i-1},v_{i-1},w_{i-1},u_i,v_i,w_i),
\\&X_1=constant.\end{align*}
We then have
\begin{align*}&I(U^nW^n;Y^n)=H(U^nW^n)-H(U^nW^n|Y^n)=\\&\sum_{i=1}^nH(U_iW_i)-H(U_iW_i|U^{i-1}W^{i-1}Y^n)
=\\&\sum_{i=1}^nI(U_iW_i;U^{i-1}W^{i-1}Y^n)\geq
\sum_{i=2}^nI(U_iW_i;U_{i-1}W_{i-1}Y_iY_{i-1})\\&=(n-1)I(U_2W_2;Y_1Y_2U_1W_1).
\end{align*}
Similarly $I(V^nW^n;Z^n)\geq (n-1)I(V_2W_2;V_1W_1Z_1Z_2)$. Next,
note that
\begin{align*}&I(V^n;Z^n|W^n)=H(V^n|W^n)-H(V^n|W^nZ^n)=\\&\sum_{i=1}^nH(V_i|W_i)-H(V_i|V^{i-1}W^nZ^n)
=\\&\sum_{i=1}^nI(V_i;V^{i-1}W^nZ^n|W_i)\geq\\&
\sum_{i=2}^nI(V_i;V_{i-1}W_{i-1}Z_iZ_{i-1}|W_i)=\\&(n-1)I(V_2;V_{1}W_{1}Z_2Z_1|W_2).
\end{align*}
Similarly, $I(U^n;Y^n|W^n)\geq (n-1)I(U_2;Y_1Y_2U_1W_1|W_2).$
Lastly, note that $I(U^n;V^n|W^n)=n\cdot I(U;V|W).$ We obtain the
desired result by substituting these values into equations
(\ref{eqna1})-(\ref{eqna5}), and letting $n\rightarrow \infty$.
\end{proof}

\section*{Acknowledgment}
This research was done when Amin Gohari was a graduate student at UC Berkeley. The research was partially supported by NSF grants CCF-0500234, CCF-0635372, CNS-0627161,
CNS-0910702, by the NSF Science \& Technology Center
grant CCF-0939370, ``Science of Information", and ARO MURI grant
W911NF-08-1-0233 ``Tools for the Analysis and Design of Complex Multi-Scale Networks."

\appendices
\section{}\label{Appndx:B} Suppose $p_0(u,v,w,x)$ is a joint distribution that maximizes $\lambda I(W;Y)+(1-\lambda)I(W;Z)+I(U;Y|W)+I(V;Z|W)-I(U;V|W)$, and among all such joint distributions has the largest value of $I(W;Y)+I(W;Z)$. In this appendix, we prove that one can find $p(\widehat{u},\widehat{v},\widehat{w},\widehat{x})$ such that
\begin{itemize}
  \item $\lambda I(W;Y)+(1-\lambda)I(W;Z)+I(U;Y|W)+I(V;Z|W)-I(U;V|W)$ is equal to $\lambda I(\widehat{W};\widehat{Y})+(1-\lambda)I(\widehat{W};\widehat{Z})+I(\widehat{U};\widehat{Y}|\widehat{W})+I(\widehat{V};\widehat{Z}|\widehat{W})-I(\widehat{U};\widehat{V}|\widehat{W})$,
  \item $I(W;Y)+I(W;Z)$ is equal to $ I(\widehat{W};\widehat{Y})+I(\widehat{W};\widehat{Z})$,
  \item $|\mathcal{\widehat{U}}|\leq \min(|\mathcal{X}|, |\mathcal{Y}|),$
  \item $|\mathcal{\widehat{V}}|\leq \min(|\mathcal{X}|, |\mathcal{Z}|),$
  \item $|\mathcal{\widehat{W}}|\leq |\mathcal{X}|,$
  \item $H(\widehat{X}|\widehat{U}\widehat{V}\widehat{W})=0$.
\end{itemize}
We begin by reducing the cardinality of $W$. Assume that $|\mathcal{W}|> |\mathcal{X}|$ and $p(w)\neq 0$ for all $w$. There must therefore exists a function $L:\mathcal{W}\rightarrow \mathbb{R}$ where
$$\mathbb{E}[L(W)|X]=0,$$
$$\exists w: p(w)\neq 0,\ \ L(w)\neq 0.$$
Let us perturb $p_0(u,v,w,x)$ along $L$ as follows:
\begin{flalign*}&p_{\epsilon}(u,v,w,x,y,z)=
p_0(u,v,w,x,y,z)\cdot [1+\epsilon
L(w)],\end{flalign*}
where $\epsilon$ is a real number in some interval
$[-\overline{\epsilon}_1, \overline{\epsilon}_2]$ for some positive reals $\overline{\epsilon}_1$ and $\overline{\epsilon}_2$.

Consider the expression $\lambda I(W;Y)+(1-\lambda)I(W;Z)+I(U;Y|W)+I(V;Z|W)-I(U;V|W)$ at $p_{\epsilon}(u,v,w,x,y,z)$. It can be verified that the expression is a linear function of $\epsilon$ under this perturbation. Since a maximum of this expression occurs at $\epsilon=0$, which is a point strictly inside the interval $[-\overline{\epsilon}_1, \overline{\epsilon}_2]$, it must be the case that this expression is a constant function of $\epsilon$. Next consider the expression $I(W;Y)+I(W;Z)$ at $p_{\epsilon}(u,v,w,x,y,z)$. It can be verified that the expression is a linear function of $\epsilon$ under this perturbation. Note that $p_0(u,v,w,x)$ is a joint distribution that has the largest value of $I(W;Y)+I(W;Z)$ among all joint distributions that maximize $\lambda I(W;Y)+(1-\lambda)I(W;Z)+I(U;Y|W)+I(V;Z|W)-I(U;V|W)$. Thus a maximum of $I(W;Y)+I(W;Z)$ occurs at $\epsilon=0$, which is a point strictly inside the interval $[-\overline{\epsilon}_1, \overline{\epsilon}_2]$. But this can only happen when $I(W;Y)+I(W;Z)$ is a constant function of $\epsilon$. Now, taking $\epsilon=-\overline{\epsilon}_1$ or $\epsilon=\overline{\epsilon}_2$ gives us a joint distribution with the same values of $\lambda I(W;Y)+(1-\lambda)I(W;Z)+I(U;Y|W)+I(V;Z|W)-I(U;V|W)$ and $I(W;Y)+I(W;Z)$, but with a smaller support on $\mathcal{W}$. Using this argument, one can reduce the cardinality of $W$ to $|\mathcal{X}|$.

Next, we show how one can reduce the cardinality of $U$ to find $p(\widehat{u},\widehat{v},\widehat{w},\widehat{x})$ such that
\begin{itemize}
  \item $\lambda I(W;Y)+(1-\lambda)I(W;Z)+I(U;Y|W)+I(V;Z|W)-I(U;V|W)$ is equal to $\lambda I(\widehat{W};\widehat{Y})+(1-\lambda)I(\widehat{W};\widehat{Z})+I(\widehat{U};\widehat{Y}|\widehat{W})+I(\widehat{V};\widehat{Z}|\widehat{W})-I(\widehat{U};\widehat{V}|\widehat{W})$,
  \item $I(W;Y)+I(W;Z)$ is equal to $ I(\widehat{W};\widehat{Y})+I(\widehat{W};\widehat{Z})$,
  \item $|\mathcal{\widehat{U}}|\leq \min(|\mathcal{X}|, |\mathcal{Y}|),$
  \item $|\mathcal{\widehat{W}}|\leq |\mathcal{X}|.$
\end{itemize}
We can repeat a similar procedure to impose the constraint $|\mathcal{\widehat{V}}|\leq \min(|\mathcal{X}|, |\mathcal{Z}|)$. Imposing the extra constraint $H(\widehat{X}|\widehat{U}\widehat{V}\widehat{W})=0$ will be discussed at the end.

If $|\mathcal{X}|\leq |\mathcal{Y}|$, establishing the cardinality bound of $|\mathcal{X}|$ on $U$ suffices. This cardinality bound is proved in Theorem 1 of \cite{EvaluationMarton}. This cardinality bound can be shown using perturbations of the type $L:\mathcal{U} \times
\mathcal{W} \rightarrow \mathbb{R}$ where
$$\mathbb{E}[L(U,W)|WX]=0.$$
Note that these perturbations preserve the marginal distribution of $p(w,x)$, and thus also $I(W;Y)+I(W;Z)$. The interesting case is therefore when $|\mathcal{X}|> |\mathcal{Y}|$. Assume that $|\mathcal{U}|> |\mathcal{Y}|.$ If for every $w\in \mathcal{W}$, $p(u|w)\neq 0$ for at most $|\mathcal{Y}|$ elements $u$, we are done, since we can relabel the elements in the range of $U$ to ensure that only an alphabet of size at most $|{\cal Y}|$ is used, without affecting any of the mutual information terms in the expression of interest. There must therefore exists a function $L:\mathcal{U} \times
\mathcal{W} \rightarrow \mathbb{R}$ where
$$\mathbb{E}[L(U,W)|WY]=0,$$
$$\exists (u,w): p_0(u,w)\neq 0,\ \ L(u,w)\neq 0.$$
 Let us perturb $p_0(u,v,w,x)$ along the random variable $L:\mathcal{U} \times
\mathcal{W} \rightarrow \mathbb{R}$. Random variables $\widetilde{U},
\widetilde{V}, \widetilde{W}, \widetilde{X}, \widetilde{Y},
\widetilde{Z}$ are distributed according to
$p_{\epsilon}(\widetilde{u}, \widetilde{v}, \widetilde{w},
\widetilde{x}, \widetilde{y}, \widetilde{z})$ defined as follows
\begin{flalign*}&p_{\epsilon}(\widetilde{u}, \widetilde{v},
\widetilde{w}, \widetilde{x}, \widetilde{y}, \widetilde{z})=
p_0(\widetilde{u}, \widetilde{v}, \widetilde{w}, \widetilde{x},
\widetilde{y}, \widetilde{z})\cdot [1+\epsilon
L(\widetilde{u},\widetilde{w})],\end{flalign*}
where $\epsilon$ is a real number in some interval
$[-\overline{\epsilon}_1, \overline{\epsilon}_2]$.

The first derivative of $\lambda I(W;Y)+(1-\lambda)I(W;Z)+I(U;Y|W)+I(V;Z|W)-I(U;V|W)$ with respect to $\epsilon$, at $\epsilon=0$ should be zero. Since \begin{eqnarray*}&\lambda I(W;Y)+(1-\lambda)I(W;Z)+I(U;Y|W)\\&+I(V;Z|W)-I(U;V|W)=\\&\lambda \big(H(W)+H(Y)-H(WY)\big)\\&+(1-\lambda)\big(H(W)+H(Z)-H(WZ)\big)+\\& H(YW)+H(ZW)-H(UYW)\\&-H(VZW)+H(UVW)-H(W),\end{eqnarray*}
we will have:
\begin{eqnarray*}&
\lambda \big(H_L(W)+H_L(Y)-H_L(WY)\big)+\\&(1-\lambda)\big(H_L(W)+H_L(Z)-H_L(WZ)\big)\\&+H_L(YW)+H_L(ZW)-H_L(UYW)\\&-H_L(VZW)+H_L(UVW)-H_L(W)=0,
\end{eqnarray*}
where $H_L(W)$ denotes
$\sum_w E[L|W=w] p(w) \log \frac{1}{p(w)}$ and similarly for the other terms.
Using Lemma 2 of \cite{EvaluationMarton}, we have:
\begin{eqnarray*}&\lambda I(\widetilde{W};\widetilde{Y})+(1-\lambda)I(\widetilde{W};\widetilde{Z})+I(\widetilde{U};\widetilde{Y}|\widetilde{W})\\&+I(\widetilde{V};\widetilde{Z}|\widetilde{W})
-I(\widetilde{U};\widetilde{V}|\widetilde{W})=\\&\lambda I(W;Y)+(1-\lambda)I(W;Z)+I(U;Y|W)\\&+I(V;Z|W)-I(U;V|W)+\\&\lambda \big(-\mathbb{E}\big[r\big(\epsilon\cdot
\mathbb{E}[L|W]\big)\big]-\mathbb{E}\big[r\big(\epsilon\cdot
\mathbb{E}[L|Y]\big)\big]\\&+\mathbb{E}\big[r\big(\epsilon\cdot
\mathbb{E}[L|WY]\big)\big]\big)+\\&(1-\lambda)\big(-\mathbb{E}\big[r\big(\epsilon\cdot
\mathbb{E}[L|W]\big)\big]-\mathbb{E}\big[r\big(\epsilon\cdot
\mathbb{E}[L|Z]\big)\big]\\&+\mathbb{E}\big[r\big(\epsilon\cdot
\mathbb{E}[L|WZ]\big)\big]\big)+\\&-\mathbb{E}\big[r\big(\epsilon\cdot
\mathbb{E}[L|YW]\big)\big]-\mathbb{E}\big[r\big(\epsilon\cdot
\mathbb{E}[L|ZW]\big)\big]\\&+
\mathbb{E}\big[r\big(\epsilon\cdot
\mathbb{E}[L|UYW]\big)\big]+\mathbb{E}\big[r\big(\epsilon\cdot
\mathbb{E}[L|VWZ]\big)\big]\\&-\mathbb{E}\big[r\big(\epsilon\cdot
\mathbb{E}[L|UVW]\big)\big]+\mathbb{E}\big[r\big(\epsilon\cdot
\mathbb{E}[L|W]\big)\big],\end{eqnarray*}
 where $r(x)=(1+x)\log(1+x)$. Since $\mathbb{E}[L(U,W)|WY]=0$, and $L$ is a function of $UW$, we have:
\begin{eqnarray*}&\lambda I(\widetilde{W};\widetilde{Y})+(1-\lambda)I(\widetilde{W};\widetilde{Z})+I(\widetilde{U};\widetilde{Y}|\widetilde{W})\\&+I(\widetilde{V};\widetilde{Z}|\widetilde{W})
-I(\widetilde{U};\widetilde{V}|\widetilde{W})=\\&\lambda I(W;Y)+(1-\lambda)I(W;Z)+I(U;Y|W)\\&+I(V;Z|W)-I(U;V|W)+\\&(1-\lambda)\big(-\mathbb{E}\big[r\big(\epsilon\cdot
\mathbb{E}[L|Z]\big)\big]+\mathbb{E}\big[r\big(\epsilon\cdot
\mathbb{E}[L|WZ]\big)\big]\big)\\&-\mathbb{E}\big[r\big(\epsilon\cdot
\mathbb{E}[L|ZW]\big)\big]+\mathbb{E}\big[r\big(\epsilon\cdot
\mathbb{E}[L|VWZ]\big)\big].\end{eqnarray*}
Since $r(x)=(1+x)\log(1+x)$ is a convex function, we have
\begin{eqnarray*}&-\mathbb{E}\big[r\big(\epsilon\cdot
\mathbb{E}[L|Z]\big)\big]+\mathbb{E}\big[r\big(\epsilon\cdot
\mathbb{E}[L|WZ]\big)\big]\geq 0,\\&-\mathbb{E}\big[r\big(\epsilon\cdot
\mathbb{E}[L|WZ]\big)\big]+\mathbb{E}\big[r\big(\epsilon\cdot
\mathbb{E}[L|VWZ]\big)\big]\geq 0.\end{eqnarray*}
 Therefore for any $\epsilon\in [-\overline{\epsilon}_1, \overline{\epsilon}_2]$, we have
\begin{eqnarray*}&\lambda I(\widetilde{W};\widetilde{Y})+(1-\lambda)I(\widetilde{W};\widetilde{Z})+I(\widetilde{U};\widetilde{Y}|\widetilde{W})\\&+I(\widetilde{V};\widetilde{Z}|\widetilde{W})
-I(\widetilde{U};\widetilde{V}|\widetilde{W})\geq\\&\lambda I(W;Y)+(1-\lambda)I(W;Z)+I(U;Y|W)\\&+I(V;Z|W)-I(U;V|W).\end{eqnarray*}
This implies that $\lambda I(\widetilde{W};\widetilde{Y})+(1-\lambda)I(\widetilde{W};\widetilde{Z})+I(\widetilde{U};\widetilde{Y}|\widetilde{W})+I(\widetilde{V};\widetilde{Z}|\widetilde{W})- I(\widetilde{U}; \widetilde{V}|\widetilde{W})$ is a constant function of $\epsilon$. The maximum of $I(\widetilde{W};\widetilde{Y})+I(\widetilde{W};\widetilde{Z})$ as a function of $\epsilon$ occurs at $\epsilon=0$. Therefore
\begin{eqnarray*}&
I_L(W;Y)+I_L(W;Z)=0,
\end{eqnarray*}
where $I_L(W;Y)$ denotes $\sum_{u,w,y}p(u,w,y)L(u,w)\log\frac{p(w,y)}{p(w)p(y)}$, etc.\ (see Lemma 2 of \cite{EvaluationMarton}).

Using Lemma 2 of \cite{EvaluationMarton}, one can observe that
$[I(\widetilde{W};\widetilde{Y}) +
I(\widetilde{W};\widetilde{Z})] - [I(W;Y) + I(W;Z)]$
equals \begin{eqnarray*}&-\mathbb{E}\big[r\big(\epsilon\cdot
\mathbb{E}[L|Z]\big)\big]+\mathbb{E}\big[r\big(\epsilon\cdot
\mathbb{E}[L|WZ]\big)\big]\geq 0.\end{eqnarray*}
  But this can only happen when $I(\widetilde{W};\widetilde{Y})+I(\widetilde{W};\widetilde{Z})$ is a constant function of $\epsilon$. Now, taking $\epsilon=-\overline{\epsilon}_1$ or $\epsilon=\overline{\epsilon}_2$ gives us auxiliary random variable $(\widetilde{U}, \widetilde{W})$ with smaller support than that of $(U,W)$. We can continue this process as long as there exists  $w\in \mathcal{W}$, such that $p(u|w)\neq 0$ for more than $|\mathcal{Y}|$ elements $u$.

It remains to show that one can impose the extra constraint $H(\widehat{X}|\widehat{U}\widehat{V}\widehat{W})=0$.
Fix $p({u},{v},{w})$. Consider the expressions $\lambda I({W};{Y})+(1-\lambda)I({W};{Z})+I({U};{Y}|{W})+I({V};{Z}|{W})-I({U};{V}|{W})$
and $I({W};{Y})+I({W};{Z})$ as functions of the conditional distribution of $r({x}|{u},{v},{w})$. We know that for instance that the former expression is maximized at $p({x}|{u},{v},{w})$. Further, the extreme points of the corresponding region for $r({x}|{u},{v},{w})$ satisfy $r({x}|{u},{v},{w})\in \{0,1\}$. Both of the expressions are convex functions of $r({x}|{u},{v},{w})$. This is
because $I({W};{Y})$ is convex in the conditional distribution
$p({y}|{w})$; similarly $I({U};{Y}|{W}={w})$ is convex for any fixed value of
${w}$. The term $I({U};{V}|{W})$ that appears with a negative sign is
constant since the joint distribution of $p({u},{v},{w})$ is fixed.

We can express $p({x}|{u},{v},{w})$ as a linear combination of the extreme points of the region formed by all conditional distributions $r({x}|{u},{v},{w})$. Since the maximum of
$\lambda I({W};{Y})+(1-\lambda)I({W};{Z})+I({U};{Y}|{W})+I({V};{Z}|{W})-I({U};{V}|{W})$
occurs at some $p({x}|{u},{v},{w})$ and the expression is convex in $r({x}|{u},{v},{w})$, the maximum must also occur at all the extreme points showing up in the linear combination. One can use the convexity of $I({W};{Y})+I({W};{Z})$ in $r({x}|{u},{v},{w})$ to show that the value of $I({W};{Y})+I({W};{Z})$ at all these extreme points must be also equal to that at $p({x}|{u},{v},{w})$.

\section{}\label{Appndx:E}
In this Appendix we complete the proof of Theorem \ref{Thm:EB1} by proving that $p(\widetilde{X}=x|U=u_i)=p(X=x|U=u_i)$ for all $i=1,2,3,...,|\mathcal{U}|$ and $x$; and similarly $p(\widetilde{X}=x|V=v_j)=p(X=x|V=v_j)$ for all $j=1,2,3,...,|\mathcal{V}|$ and $x$.

Note that \small
\begin{eqnarray*}&p(\widetilde{X}=x|U=u_i)=\\&\sum_{j}p(V=v_j|U=u_i)p(\widetilde{X}=x|U=u_i,V=v_j)=\\&\sum_{j}p(V=v_j|U=u_i)\sum_{k=0}^M p(T_{i,j}=k)\textbf{1}[\xi_{k}(u_i,v_j)=x]=\\&
\sum_{j}p(V=v_j|U=u_i)(1-\frac{\epsilon}{\pi_{i,j}})\textbf{1}[\xi_{0}(u_i,v_j)=x]+\\&
\sum_{j}p(V=v_j|U=u_i)\sum_{k=1}^M \frac{\epsilon}{\pi_{i,j}}\alpha_k\textbf{1}[\xi_{k}(u_i,v_j)=x]=\\&
\sum_{j}p(V=v_j|U=u_i)(\frac{\pi_{i,j}-\epsilon}{\pi_{i,j}})\textbf{1}[\xi_{0}(u_i,v_j)=x]+\\&
\sum_{k=1}^M\sum_{j}p(V=v_j|U=u_i)\frac{\epsilon}{\pi_{i,j}}\alpha_k\textbf{1}[\xi_{k}(u_i,v_j)=x].\\&
\end{eqnarray*}\normalsize
Note that $p(V=v_j|U=u_i)=\frac{p(V=v_j,U=u_i)}{p(U=u_i)}=\frac{\pi_{i,j}}{p(U=u_i)}$. Therefore
\begin{eqnarray*}&p(\widetilde{X}=x|U=u_i)=\\&
\sum_{j}\frac{\pi_{i,j}-\epsilon}{p(U=u_i)}\textbf{1}[\xi_{0}(u_i,v_j)=x]+\\&
\sum_{k=1}^M\sum_{j}\frac{\epsilon}{p(U=u_i)}\alpha_k\textbf{1}[\xi_{k}(u_i,v_j)=x]=\\&
\sum_{j}\frac{\pi_{i,j}}{p(U=u_i)}\textbf{1}[\xi_{0}(u_i,v_j)=x]\\&-\frac{\epsilon}{p(U=u_i)}\sum_{j}\textbf{1}[\xi_{0}(u_i,v_j)=x]+\\&
\frac{\epsilon}{p(U=u_i)}\sum_{k=1}^M\alpha_k\sum_{j}\textbf{1}[\xi_{k}(u_i,v_j)=x].
\end{eqnarray*}
But since
$$\overrightarrow{v_{\xi_0}}=\sum_{t=1}^{M}\alpha_t\overrightarrow{v_{\xi_t}},$$
the profiles of the $i^{th}$ rows must also satisfy the same property:
$$\sum_{j}\textbf{1}[\xi_{0}(u_i,v_j)=x]=\sum_{k=1}^M\alpha_k\sum_{j}\textbf{1}[\xi_{k}(u_i,v_j)=x].$$
Therefore,
\begin{eqnarray*}&p(\widetilde{X}=x|U=u_i)=\\&
\sum_{j}\frac{\pi_{i,j}}{p(U=u_i)}\textbf{1}[\xi_{0}(u_i,v_j)=x]+0-0=\\&
\sum_{j}\frac{\pi_{i,j}}{p(U=u_i)}\textbf{1}[\xi_{0}(u_i,v_j)=x]=p(X=x|U=u_i).
\end{eqnarray*}

The equation $p(\widetilde{X}=x|V=v_j)=p(X=x|V=v_j)$ for all $j=1,2,3,...,|\mathcal{V}|$ and $x$ can be proved similarly.

\section{}\label{Appndx:C01}
Note that \small\begin{eqnarray*}&p(\widetilde{X}=x_0|U=u_i, T_{i,j}=0)=\\&
p(\widetilde{X}=x_0|U=u_i, T_{i,j}=0,V=v_j)p(V=v_j|U=u_i, T_{i,j}=0)\\&+\\&
p(\widetilde{X}=x_0|U=u_i, T_{i,j}=0,V \neq v_j)p(V \neq v_j|U=u_i, T_{i,j}=0).\end{eqnarray*}\normalsize
Since under the event $(U,V)=(u_i,v_j)$ and $T_{i,j}=0$, $\widetilde{X}$ is equal to $x_0$, the term $p(\widetilde{X}=x_0|U=u_i, T_{i,j}=0,V=v_j)$ will be equal to one. Since $(U,V)$ is independent of $T_{i,j}$, we have $$p(V=v_j|U=u_i, T_{i,j}=0)=p(V=v_j|U=u_i),$$$$p(V \neq v_j|U=u_i, T_{i,j}=0)=p(V \neq v_j|U=u_i).$$ Lastly $p(\widetilde{X}=x_0|U=u_i, T_{i,j}=0,V \neq v_j)$ is equal to $p(\widetilde{X}=x_0|U=u_i, V \neq v_j)$ since under the event that $(U=u_i, V \neq v_j)$, $\widetilde{X}$ will be independent of $T_{i,j}$ (note that $T_{\cdot, \cdot}$ random variables were mutually independent of each other). Therefore,
\begin{eqnarray}&p(\widetilde{X}=x_0|U=u_i, T_{i,j}=0)=\label{eqn:AA1}\\&
p(V=v_j|U=u_i)+\nonumber\\&
p(\widetilde{X}=x_0|U=u_i, V \neq v_j)p(V \neq v_j|U=u_i).\nonumber\end{eqnarray}
Next, note that \small\begin{align*}&p(\widetilde{X}=x_0|U=u_i, T_{i,j}=1)=\\&
p(\widetilde{X}=x_0|U=u_i, T_{i,j}=1,V=v_j)p(V=v_j|U=u_i, T_{i,j}=1)+\\&
p(\widetilde{X}=x_0|U=u_i, T_{i,j}=1,V \neq v_j)p(V \neq v_j|U=u_i, T_{i,j}=1).\end{align*}\normalsize
Since under the event $(U,V)=(u_i,v_j)$ and $T_{i,j}=1$, $\widetilde{X}$ is equal to $x_1$, the term $p(\widetilde{X}=x_0|U=u_i, T_{i,j}=1,V=v_j)$ will be equal to zero. Following an argument like above, one can show that
\begin{eqnarray}&p(\widetilde{X}=x_0|U=u_i, T_{i,j}=1)=\label{eqn:AA2}\\&
0+p(\widetilde{X}=x_0|U=u_i, V \neq v_j)p(V \neq v_j|U=u_i).\nonumber\end{eqnarray}
Comparing equations (\ref{eqn:AA1}) and (\ref{eqn:AA2}), and noting that $p(V=v_j|U=u_i)>0$, we conclude that
$$p(\widetilde{X}=x_0|U=u_i, T_{i,j}=0)
\neq p(\widetilde{X}=x_0|U=u_i, T_{i,j}=1).$$

The proof for $$p(\widetilde{X}=x_1|U=u_i, T_{i,j}=0)
\neq p(\widetilde{X}=x_1|U=u_i, T_{i,j}=1)$$
 is similar.

It remains to show that for any $x\notin\{x_0,x_1\}$, $$p(\widetilde{X}=x|U=u_i, T_{i,j}=0)
= p(\widetilde{X}=x|U=u_i, T_{i,j}=1).$$
Note that
\small\begin{align*}&p(\widetilde{X}=x|U=u_i, T_{i,j}=1)=\\&
p(\widetilde{X}=x|U=u_i, T_{i,j}=1,V=v_j)p(V=v_j|U=u_i, T_{i,j}=1)+\\&
p(\widetilde{X}=x|U=u_i, T_{i,j}=1,V \neq v_j)p(V \neq v_j|U=u_i, T_{i,j}=1)=\\&
0+p(\widetilde{X}=x|U=u_i, V \neq v_j)p(V \neq v_j|U=u_i)=\\&p(\widetilde{X}=x|U=u_i, T_{i,j}=0).\end{align*}
\normalsize
\section{}\label{Appndx:C02}
We prove the statement by contradiction. Assume that
$$p(\widetilde{Y}=y|U=u_i, T_{i,j}=0)
=p(\widetilde{Y}=y|U=u_i, T_{i,j}=1).$$

We have
\begin{align*}&p(\widetilde{Y}=y|U=u_i, T_{i,j}=0)=\\&
p(\widetilde{Y}=y|U=u_i, T_{i,j}=0, \widetilde{X}=x_0)p(\widetilde{X}=x_0|U=u_i, T_{i,j}=0)+\\&
p(\widetilde{Y}=y|U=u_i, T_{i,j}=0, \widetilde{X}=x_1)p(\widetilde{X}=x_1|U=u_i, T_{i,j}=0)+\\&
\sum_{x\in \mathcal{X}, x\notin \{x_0,x_1\}}\big( p(\widetilde{Y}=y|U=u_i, T_{i,j}=0, \widetilde{X}=x)\times\\&\quad\quad\quad\quad\quad\quad p(\widetilde{X}=x|U=u_i, T_{i,j}=0)\big)=\\&
p(\widetilde{Y}=y|\widetilde{X}=x_0)p(\widetilde{X}=x_0|U=u_i, T_{i,j}=0)+\\&
p(\widetilde{Y}=y|\widetilde{X}=x_1)p(\widetilde{X}=x_1|U=u_i, T_{i,j}=0)+\\&
\sum_{x \mathcal{X}, x\notin \{x_0,x_1\}}\big( p(\widetilde{Y}=y|\widetilde{X}=x)p(\widetilde{X}=x|U=u_i, T_{i,j}=0)\big).\end{align*}
Similarly,
\begin{align*}&p(\widetilde{Y}=y|U=u_i, T_{i,j}=1)=\\&
p(\widetilde{Y}=y|\widetilde{X}=x_0)p(\widetilde{X}=x_0|U=u_i, T_{i,j}=1)+\\&
p(\widetilde{Y}=y|\widetilde{X}=x_1)p(\widetilde{X}=x_1|U=u_i, T_{i,j}=1)+\\&
\sum_{x \mathcal{X}, x\notin \{x_0,x_1\}}\big( p(\widetilde{Y}=y|\widetilde{X}=x)p(\widetilde{X}=x|U=u_i, T_{i,j}=1)\big).\end{align*}
%The assumption that $$p(\widetilde{Y}=y|U=u_i, T_{i,j}=0)
%=p(\widetilde{Y}=y|U=u_i, T_{i,j}=1),$$ therefore implies:
%\begin{align*}&p(\widetilde{Y}=y|\widetilde{X}=x_0)p(\widetilde{X}=x_0|U=u_i, T_{i,j}=0)+\\&
%p(\widetilde{Y}=y|\widetilde{X}=x_1)p(\widetilde{X}=x_1|U=u_i, T_{i,j}=0)+\\&
%\sum_{x \mathcal{X}, x\notin \{x_0,x_1\}}\big( p(\widetilde{Y}=y|\widetilde{X}=x)\times\\&p(\widetilde{X}=x|U=u_i, T_{i,j}=0)\big)=\\&
%p(\widetilde{Y}=y|\widetilde{X}=x_0)p(\widetilde{X}=x_0|U=u_i, T_{i,j}=1)+\\&
%p(\widetilde{Y}=y|\widetilde{X}=x_1)p(\widetilde{X}=x_1|U=u_i, T_{i,j}=1)+\\&
%\sum_{x\in \mathcal{X}, x\notin \{x_0,x_1\}}\big( p(\widetilde{Y}=y|\widetilde{X}=x)p(\widetilde{X}=x|U=u_i, T_{i,j}=1)\big).\end{align*}
It was shown in Appendix \ref{Appndx:C01} that
$$p(\widetilde{X}=x_0|U=u_i, T_{i,j}=0)
\neq p(\widetilde{X}=x_0|U=u_i, T_{i,j}=1),$$
$$p(\widetilde{X}=x_1|U=u_i, T_{i,j}=0)\neq p(\widetilde{X}=x_1|U=u_i, T_{i,j}=1).$$ But for any $x\notin\{x_0,x_1\}$,
\begin{eqnarray}&p(\widetilde{X}=x|U=u_i, T_{i,j}=0)=\label{eqn:A3}\\& p(\widetilde{X}=x|U=u_i, T_{i,j}=1).\nonumber\end{eqnarray}
Thus, we must have
\begin{eqnarray*}&p(\widetilde{Y}=y|\widetilde{X}=x_0)p(\widetilde{X}=x_0|U=u_i, T_{i,j}=0)+\\&
p(\widetilde{Y}=y|\widetilde{X}=x_1)p(\widetilde{X}=x_1|U=u_i, T_{i,j}=0)=\\&
p(\widetilde{Y}=y|\widetilde{X}=x_0)p(\widetilde{X}=x_0|U=u_i, T_{i,j}=1)+\\&
p(\widetilde{Y}=y|\widetilde{X}=x_1)p(\widetilde{X}=x_1|U=u_i, T_{i,j}=1).\end{eqnarray*}
This implies that \begin{eqnarray*}&\frac{p(\widetilde{X}=x_0|U=u_i, T_{i,j}=1)-p(\widetilde{X}=x_0|U=u_i, T_{i,j}=0)}{
p(\widetilde{X}=x_1|U=u_i, T_{i,j}=0)-p(\widetilde{X}=x_1|U=u_i, T_{i,j}=1)}=
\frac{p(\widetilde{Y}=y|\widetilde{X}=x_1)}{p(\widetilde{Y}=y|\widetilde{X}=x_0)}.\end{eqnarray*}
Note that the nominator and denominator are positive by what was proved in Appendix \ref{Appndx:C01}.

On the other hand, we also have by equation (\ref{eqn:A3}):
\begin{eqnarray*}&p(\widetilde{X}=x_0|U=u_i, T_{i,j}=0)+\\&
p(\widetilde{X}=x_1|U=u_i, T_{i,j}=0)=\\&
p(\widetilde{X}=x_0|U=u_i, T_{i,j}=1)+\\&
p(\widetilde{X}=x_1|U=u_i, T_{i,j}=1).\end{eqnarray*}
This implies that \begin{eqnarray*}&\frac{p(\widetilde{X}=x_0|U=u_i, T_{i,j}=1)-p(\widetilde{X}=x_0|U=u_i, T_{i,j}=0)}{
p(\widetilde{X}=x_1|U=u_i, T_{i,j}=0)-p(\widetilde{X}=x_1|U=u_i, T_{i,j}=1)}=1.\end{eqnarray*}
Hence, \begin{eqnarray*}&\frac{p(\widetilde{Y}=y|\widetilde{X}=x_1)}{p(\widetilde{Y}=y|\widetilde{X}=x_0)}=1.\end{eqnarray*} But we know that $p(\widetilde{Y}=y|\widetilde{X}=x_0)\neq p(\widetilde{Y}=y|\widetilde{X}=x_1)$ since the input values $x_0$ and $x_1$ are distinguishable by the $Y$ receiver. This is a contradiction.

\section{}\label{Appndx:D} The proof follows from the following two statements:

\emph{Statement 1: } Assume that $p^*(u,v,w,x)$ is an arbitrary joint distribution maximizing $\lambda I(W;Y)+ (1-\lambda)
I(W;Z)+I(U;Y|W)+I(V;Z|W)-I(U;V|W)$, and having the largest value of $I(W;Y)+I(W;Z)$ among all maximizing joint distributions. For every $w$, $p^*(x|w)$ must belong to the set
$\mathcal{T}(q(y,z|x))$ defined as follows. Let $\mathcal{T}(q(y,z|x))$ be the set of pmfs on $\mathcal{X}$, $t(x)$, such that
\begin{align*}
&\max_{p(u,v,w|x)t(x)q(y,z|x)}\big\{ \lambda I(W;Y)+ (1-\lambda)
I(W;Z)\\
&\quad + I(U;Y|W)+I(V;Z|W)-I(U;V|W)\big\}\\
&\quad = \max_{p(u,v|x)t(x)q(y,z|x)} (I(U;Y)+I(V;Z)-I(U;V)),
\end{align*}
and $I(W;Y)=I(W;Z)=0$ for
\emph{any}\footnote{Note that such a pmf may not unique.} pmf $p(u,v,w|x)t(x)$ that maximizes the expression
$\lambda I(W;Y)+ (1-\lambda)
I(W;Z)+I(U;Y|W)+I(V;Z|W)-I(U;V|W)$.\footnote{We
have used maximum and not supremum in the above conditions since
cardinality bounds on the auxiliary random variables exist
\cite{EvaluationMarton}.}

\emph{Statement 2: } Let $q(y,z|x)$ be a general broadcast channel, and $t(x)\in \mathcal{T}(q(y,z|x))$. Consider the maximization problem: $\max_{p(u,v|x)t(x)q(y,z|x)} (I(U;Y)+I(V;Z)-I(U;V))$. Assume that \emph{a} maximum occurs at $p^*(u,v|x)$. Then the following holds for random variables $(U, V, X, Y, Z) \sim p^*(u,v|x)t(x)q(y,z|x)$:
\begin{itemize}
  \item $I(\overline{U};Y)\geq I(\overline{U};VZ)$ for every $\overline{U}\rightarrow U \rightarrow VXYZ$.
  \item $I(\overline{V};Z)\geq I(\overline{V};UY)$ for every $\overline{V}\rightarrow V \rightarrow UXYZ$.
\end{itemize}

\subsection{Proof of Statement 1: }
Assume that the
marginal pmf of $X$ given $W=w$ does not belong to
$\mathcal{T}$ for some $w$. By the definition then,
at least one of the following must hold:

\emph{Case 1:} Corresponding to $p^*_{X|W=w}(x)$ is the conditional
distribution $p(\widehat{u},\widehat{v},\widehat{w}|\widehat{x})$
such that
\begin{align}&I(U;Y|W=w)+I(V;Z|W=w)-I(U;V|W=w)\nonumber<\\&\lambda I(\widehat{W};\widehat{Y})+ (1-\lambda)
I(\widehat{W};\widehat{Z}) +I(\widehat{U};\widehat{Y}|\widehat{W})\nonumber\\&\qquad +I(\widehat{V};\widehat{Z}|\widehat{W})-
I(\widehat{U};\widehat{V}|\widehat{W})\label{eqn:E3}\end{align}
where $p(\widehat{u},\widehat{v}, \widehat{w}, \widehat{x},\widehat{y}, \widehat{z})=p(\widehat{u},\widehat{v},\widehat{w}|\widehat{x})p^*_{X|W=w}(\widehat{x})q(\widehat{y},\widehat{z}|\widehat{x})$.

\emph{Case 2:} Corresponding to $p^*_{X|W=w}(x)$ is the conditional
distribution $p(\widehat{u},\widehat{v},\widehat{w}|\widehat{x})$
such that
\begin{align*}&I(U;Y|W=w)+I(V;Z|W=w)-I(U;V|W=w)\nonumber=\\&\lambda I(\widehat{W};\widehat{Y})+ (1-\lambda)
I(\widehat{W};\widehat{Z})+I(\widehat{U};\widehat{Y}|\widehat{W})\nonumber\\&\qquad +I(\widehat{V};\widehat{Z}|\widehat{W})-
I(\widehat{U};\widehat{V}|\widehat{W})\end{align*}
but $I(\widehat{W};\widehat{Y})+I(\widehat{W};\widehat{Z})>0$, where $p(\widehat{u},\widehat{v}, \widehat{w}, \widehat{x},\widehat{y}, \widehat{z})=p(\widehat{u},\widehat{v},\widehat{w}|\widehat{x})p^*_{X|W=w}(\widehat{x})q(\widehat{y},\widehat{z}|\widehat{x})$.

Define $\widetilde{U}$, $\widetilde{V}$, $\widetilde{W}$ jointly
distributed with $U$, $V$, $W$, $X$, $Y$, $Z$ as follows: whenever
$W\neq w$, the random variables $\widetilde{U}=U$, $\widetilde{V}=V$,
$\widetilde{W}=W$. For $W=w$, the
Markov chain $\widetilde{U}\widetilde{V}\widetilde{W}\rightarrow
X\rightarrow UVWYZ$ holds, and $p(\widetilde{u}, \widetilde{v},
\widetilde{w}|x)=p(\widehat{u},\widehat{v},\widehat{w}|\widehat{x})$. Next, assume
that $U'=\widetilde{U}$, $V'=\widetilde{V}$, $W'=W\widetilde{W}$.

If case 1 holds, we prove that $\lambda I(W';Y)+(1-\lambda) I(W';Z)+
I(U';Y|W')+I(V';Z|W')-I(U';V'|W')> \lambda I(W;Y)+(1-\lambda) I(W;Z)+ I(U;Y|W)+I(V;Z|W)-I(U;V|W),$ which results in a
contradiction. If case 2 holds, we prove that $\lambda I(W';Y)+(1-\lambda) I(W';Z)+ I(U';Y|W')+I(V';Z|W')-I(U';V'|W')=\lambda I(W;Y)+(1-\lambda) I(W;Z)+ I(U;Y|W)+I(V;Z|W)-I(U;V|W)$ but that
$I(W';Y)+I(W';Z)> I(W;Y)+I(W;Z)$, which results in a contradiction.

Assume that case 1 holds. Since $W'=W\widetilde{W}$, $I(W';Y)=I(W;Y)+I(\widetilde{W};Y|W)$ and $I(W';Z)=I(W;Z)+I(\widetilde{W};Z|W)$, we need to show that
\begin{eqnarray*}&\lambda I(\widetilde{W};Y|W) + (1-\lambda) I(\widetilde{W};Z|W)+ I(\widetilde{U};Y|W\widetilde{W})+\\&I(\widetilde{V};Z|W\widetilde{W})-I(\widetilde{U};\widetilde{V}|W\widetilde{W})>\\&
I(U;Y|W)+I(V;Z|W)-I(U;V|W)
\end{eqnarray*}
%We prove equation (\ref{eqn:E1}); the proof for equation (\ref{eqn:E2}) is similar. In order to prove  equation (\ref{eqn:E1}), it suffices
%to show that
%\begin{eqnarray*}&I(W';Y)+
%I(U';Y|W')+I(V';Z|W')-I(U';V'|W')\\&>I(W;Y)+
%I(U;Y|W)+I(V;Z|W)-I(U;V|W),\end{eqnarray*}
%Since $W'=W\widetilde{W}$, $I(W';Y)=I(W;Y)+I(\widetilde{W};Y|W)$, thus we need to show that
%\begin{eqnarray*}&I(\widetilde{W};Y|W)+
%I(\widetilde{U};Y|W\widetilde{W})+\\&I(\widetilde{V};Z|W\widetilde{W})-I(\widetilde{U};\widetilde{V}|W\widetilde{W})>\\&
%I(U;Y|W)+I(V;Z|W)-I(U;V|W).\end{eqnarray*}
Remember that whenever
$W\neq w$, random variables $\widetilde{U}$, $\widetilde{V}$,
$\widetilde{W}$ were defined to be equal to $U$, $V$, $W$. Therefore we need to show that
\begin{eqnarray*}&p(W=w)\big[\lambda I(\widetilde{W};Y|W=w) + (1-\lambda) I(\widetilde{W};Z|W=w)+\\& I(\widetilde{U};Y|W=w,\widetilde{W})+I(\widetilde{V};Z|W=w,\widetilde{W})\\&-I(\widetilde{U};\widetilde{V}|W=w,\widetilde{W})\big]>\\&
p(W=w)\big[I(U;Y|W=w)+I(V;Z|W=w)\\&-I(U;V|W=w)\big].
\end{eqnarray*}
On the event $W=w$, random variables $\widetilde{U}$, $\widetilde{V}$,
$\widetilde{W}$ were defined so that $p(\widetilde{u}, \widetilde{v}, \widetilde{w}|x)$ is equal to
$p(\widehat{u},\widehat{v},\widehat{w}|\widehat{x})$. Furthermore the marginal distribution of $p(\widehat{x})$ is $p^*(x|W=w)$. Therefore $I(\widetilde{W};Y|W=w)=I(\widehat{W};\widehat{Y})$, $I(\widetilde{W};Z|W=w)=I(\widehat{W};\widehat{Z})$, $I(\widetilde{U};Y|W=w,\widetilde{W})=I(\widehat{U};\widehat{Y}|\widehat{W})$, etc. Thus it remains to show that
\begin{eqnarray*}&\lambda I(\widehat{W};\widehat{Y})+(1-\lambda) I(\widehat{W};\widehat{Z})+I(\widehat{U};\widehat{Y}|\widehat{W})\\&+I(\widehat{V};\widehat{Z}|\widehat{W})-I(\widehat{U};\widehat{V}|\widehat{W})>\\&
I(U;Y|W=w)+I(V;Z|W=w)-I(U;V|W=w).\end{eqnarray*}
This holds because of equation (\ref{eqn:E3}). This concludes the proof for case 1.

Now, assume that case 2 holds. Following, the above proof for case 1, one can get \begin{eqnarray*}&\lambda I(W';Y)+(1-\lambda) I(W';Z)+
I(U';Y|W')\nonumber\\&+I(V';Z|W')-I(U';V'|W')\geq \\&\lambda I(W;Y)+(1-\lambda)I(W;Z)+
I(U;Y|W)\nonumber\\&+I(V;Z|W)-I(U;V|W).\end{eqnarray*}
Note that $I(W';Y)+I(W';Z)=I(W;Y)+I(\widetilde{W};Y|W)+I(W;Z)+I(\widetilde{W};Z|W)$. Thus, we need to show that $I(\widetilde{W};Y|W)+I(\widetilde{W};Z|W)>0$. Note that
\begin{eqnarray*}&I(\widetilde{W};Y|W)+I(\widetilde{W};Z|W)=\\&p(W=w)\big(I(\widetilde{W};Y|W=w)+I(\widetilde{W};Z|W=w)\big)\\&
=p(W=w)\big(I(\widehat{W};\widehat{Y})+I(\widehat{W};\widehat{Z})\big)>0.\end{eqnarray*}

\subsection{Proof of Statement 2: } Take an arbitrary $\overline{U}$ satisfying $\overline{U}\rightarrow U \rightarrow VXYZ$. Let $\widehat{W}=\overline{U}$, $\widehat{U}=U$, $\widehat{V}=V$. Since $t(x)\in \mathcal{T}(q(y,z|x))$, and $p^*(u,v|x)$ maximizes $I(U;Y)+I(V;Z)-I(U;V)$, we can write:
\begin{align}&I(U;Y)+I(V;Z)-I(U;V)\geq\nonumber\\& \lambda I(\widehat{W};Y)+ (1-\lambda) I(\widehat{W};Z)+I(\widehat{U};Y|\widehat{W})+I(\widehat{V};Z|\widehat{W})\nonumber\\&\qquad -I(\widehat{U};\widehat{V}|\widehat{W}),\label{eqn:lemma2a}\end{align} and furthermore if equality holds, we must have $I(\widehat{W};Y)=I(\widehat{W};Z)=0$. We prove that this implies that $I(\overline{U};Y)\geq I(\overline{U};VZ)$.

We can write:
\begin{eqnarray*}&I(U;Y)+I(V;Z)-I(U;V)\geq\\& \lambda I(\widehat{W};Y)+ (1-\lambda) I(\widehat{W};Z)+I(\widehat{U};Y|\widehat{W})+I(\widehat{V};Z|\widehat{W})\nonumber\\&-I(\widehat{U};\widehat{V}|\widehat{W})=\\&
\lambda I(\overline{U};Y)+(1-\lambda) I(\overline{U};Z)+I(U;Y|\overline{U})\nonumber\\&+I(V;Z|\overline{U})-I(U;V|\overline{U}).\end{eqnarray*}
%Therefore
%\begin{eqnarray*}&I(U;Y)+I(V;Z)-I(U;V)\geq\\&
%\lambda I(\overline{U};Y)+(1-\lambda) I(\overline{U};Z)+I(U;Y|\overline{U})+I(V;Z|\overline{U})\nonumber\\&-I(U;V|\overline{U})\end{eqnarray*}
Since $\overline{U}\rightarrow U \rightarrow VXYZ$, we have $I(U;Y)=I(\overline{U}U;Y)$ and $I(U;V)=I(\overline{U}U;V)$. This implies that
\begin{eqnarray*}&I(\overline{U};Y)+I(V;Z)-I(\overline{U};V)\geq\nonumber\\&
\lambda I(\overline{U};Y)+(1-\lambda) I(\overline{U};Z)+I(V;Z|\overline{U})\end{eqnarray*}
or,
\begin{eqnarray*}&I(\overline{U};Y)+I(V;Z)\geq
\lambda I(\overline{U};Y)+(1-\lambda) I(\overline{U};Z)+I(V;Z\overline{U})\end{eqnarray*}
or,
\begin{eqnarray*}&(1-\lambda)I(\overline{U};Y)\geq
(1-\lambda) I(\overline{U};Z)+I(V;\overline{U}|Z).\end{eqnarray*}
In other words
\begin{eqnarray}&(1-\lambda)I(\overline{U};Y)\geq
(1-\lambda) I(\overline{U};VZ)+\lambda I(V;\overline{U}|Z).\label{eqn:1}\end{eqnarray}

Let us consider the following two cases:
\begin{itemize}
  \item $\lambda<1$: In this case, equation (\ref{eqn:1}) implies that $I(\overline{U};Y)\geq I(\overline{U};VZ) + \frac{\lambda}{1-\lambda}I(V;\overline{U}|Z)$. This inequality implies the desired inequality $I(\overline{U};Y)\geq I(\overline{U};VZ)$.
  \item $\lambda=1$: In this case, equation (\ref{eqn:1}) implies that $I(V;\overline{U}|Z)=0$. Furthermore equation (\ref{eqn:lemma2a}) will hold with equality. Since $t(x)\in \mathcal{T}$, we must have $I(\overline{U};Y)=I(\overline{U};Z)=0$.$\\$ The fact that $I(V;\overline{U}|Z)=I(\overline{U};Y)=I(\overline{U};Z)=0$ implies that
      $I(\overline{U};Y)=I(\overline{U};ZV)=0$. Therefore the inequality $I(\overline{U};Y)\geq I(\overline{U};ZV)$ also holds in this case.
\end{itemize}
In each case, we are done. The proof for the inequality $I(\overline{V};Z)\geq I(\overline{V};YU)$ is similar.

\end{document}